\documentclass[12pt,a4paper,oneside]{scrartcl}

\usepackage{setspace}

\usepackage{graphicx}
\usepackage{amsmath,amssymb,amsthm}
\usepackage{url}
\usepackage{ifthen}
\usepackage{natbib}
\usepackage{color}
\usepackage{subfigure}
\usepackage{epstopdf}
\usepackage{arydshln}
\usepackage[colorlinks, linkcolor = black, citecolor = black, filecolor = black, urlcolor = blue]{hyperref}
\usepackage{soul}
\usepackage{comment}
\usepackage{multirow}
\usepackage[symbol]{footmisc}

\theoremstyle{plain}
\newtheorem{theorem}{Theorem}%
\newtheorem{corollary}[theorem]{Corollary}%
\newtheorem{proposition}[theorem]{Proposition}%
\newtheorem{lemma}[theorem]{Lemma}%
\newtheorem{remark}[theorem]{Remark}%

\usepackage{marginnote}

\usepackage[colorinlistoftodos,textsize=tiny]{todonotes}
\newcommand{\Comments}{1}
\newcommand{\mynote}[2]{\ifnum\Comments=1\textcolor{#1}{#2}\fi}
\newcommand{\mytodo}[2]{\ifnum\Comments=1%
	\todo[linecolor=#1!80!black,backgroundcolor=#1,bordercolor=#1!80!black]{#2}\fi}

\newcommand{\BK}[1]{\mytodo{green!20!white}{BK: #1}}

\ifnum\Comments=1               
\fi

\newcommand{\N}{\ensuremath{\mathbb{N}}}%
\newcommand{\Z}{\ensuremath{\mathbb{Z}}}%
\newcommand{\R}{\ensuremath{\mathbb{R}}}%
\newcommand{\C}{\ensuremath{\mathbb{C}}}%
\newcommand{\ind}{{\boldsymbol 1}}

\DeclareMathOperator*{\MISE}{\operatorname{MISE}}

\DeclareMathOperator*{\ISB}{\operatorname{ISB}}
\DeclareMathOperator*{\AMISE}{\operatorname{AMISE}}
\DeclareMathOperator*{\AVar}{\operatorname{AVar}}

\DeclareMathOperator*{\Imag}{\operatorname{Im}}

\newcommand{\cp}{\stackrel{p}{\longrightarrow}}

\begin{document}
\title{Nonparametric estimation on the circle based on Fej\'er polynomials}

\author{Bernhard Klar$^{1}$, Bojana Milo\v sevi\' c\,$^{2}$, and Marko Obradovi\' c\,$^{2}$ \\
  \small{$^{1}$ Institute of Stochastics, Karlsruhe Institute of Technology (KIT), Germany, \footnote{bernhard.klar@kit.edu}} \\
  \small{$^{2}$ Faculty of Mathematics, University of Belgrade, Serbia, \footnote{bojana@matf.bg.ac.rs, marcone@matf.bg.ac.rs}}}

\date{\today }
\maketitle

\begin{abstract}
This paper presents a comprehensive study of nonparametric estimation techniques on the circle using Fej\'er polynomials, which are analogues of Bernstein polynomials for periodic functions. Building upon Fej\'er’s uniform approximation theorem, the paper introduces circular density and distribution function estimators based on Fej\'er kernels. It establishes their theoretical properties, including uniform strong consistency and asymptotic expansions. Since the estimation of the distribution function on the circle depends on the choice of the origin, we propose a data-dependent method to address this issue.    

The proposed methods are extended to account for measurement errors by incorporating classical and Berkson error models, adjusting the Fej\'er estimator to mitigate their effects. Simulation studies analyze the finite-sample performance of these estimators under various scenarios, including mixtures of circular distributions and measurement error models. An application to rainfall data demonstrates the practical application of the proposed estimators, demonstrating their robustness and effectiveness in the presence of rounding-induced Berkson errors. 
\end{abstract}


\section{Introduction}
Only a few years after \cite{Ro:1956,Ro:1971} and \cite{PA:1962} proposed kernel estimators for nonparametric probability density estimation, \cite{VI:1975} developed an alternative method based on Bernstein polynomials. Since this seminal work, both approaches have been developed more or less independently in many subsequent papers, including extensions to nonparametric estimation of the cumulative distribution function (CDF). The idea pursued by Vitale is ultimately based on Bernstein's proof of Weierstrass' approximation theorem, in which the (unknown) density or CDF is replaced by an estimator; it turns out that the result can also be interpreted as a kernel density estimator. The situation is similar for extensions of the method to distributions on the positive real numbers, where the uniform approximation of functions is based on Poisson weights instead of binomial weights as for the Bernstein polynomials.

The purpose of this paper is to analyze the merits of a corresponding method for distributions on the circle. There are already proposals in this direction: \cite{CW:2018} proposed the use of Bernstein polynomials for density estimation on the circle; however, this does not provide a periodic estimator in general. Therefore, \cite{CH:2022} recommended a nonlinear transformation combined with the Bernstein polynomial density estimator on the unit interval to obtain a periodic density; no further theory is given. However, both proposals are adaptations of the linear case to the circle and do not seem to exploit the special situation of the periodicity of the occurring densities. 

An analogue of Bernstein’s proof of Weierstrass’ approximation theorem for the case of periodic functions can be found in Fej\'er's approximation theorem, where periodic functions are uniformly approximated by trigonometric polynomials of order $n$, the so-called Fej\'er polynomials. If the unknown quantities are replaced by empirical counterparts, nonparametric estimators on the circle are obtained naturally. Again, the dualism mentioned above is evident; the estimators can be interpreted as (integrated) kernel estimators. The resulting density estimator appears in \cite{DM:2009} but without further analysis (see comments in Section \ref{sec:asym-exp}).
We also mention that \cite{WI:1975} used Fej\'er polynomials for density estimation on the real line.

\smallskip
Regarding nonparametric density estimation on the circle with kernel-type estimators, there exist several proposals. Estimates on (hyper-)spheres ${\cal S}^{p-1}$ were proposed by \cite{HW:1987} and \cite{BR:1988}. Their proposals are tailored to the case $p\geq 3$, leading to rather strong conditions imposed on the kernels. A summary of the results can be found in \cite{LV:2017}.
Since the Fej\'er kernel does not fulfil the conditions on the kernel \citep[Remark 3]{DM:2009}, the results can not be used for the density estimation based on Fej\'er polynomials. 
\cite{DM:2009,DM:2011} provided a theoretical basis for kernel density estimators on the circle and torus. They defined a circular density estimator by
\begin{align} \label{def:cde}
\tilde{f}_n(\theta;\kappa) &= \frac{1}{n}\sum_{j=1}^n K_\kappa(\theta-X_j),
\end{align}
where $K_\kappa:[0,2\pi)\to\R$ is a circular kernel (of order $r$) with concentration parameter $\kappa>0$ satisfying certain conditions (which, as mentioned above, are not satisfied by the Fej\'er kernel). Typically, $K$ equals a circular density function with center 0.  
\cite{CH:2018} considers the special case of the wrapped Cauchy distribution, i.e. $K_\kappa(\theta-\Theta_j)=f_{WC}(\theta;\Theta_j,\kappa)$ and shows connections to certain Fourier series estimators on the circle. Another candidate is the density of the von Mises distribution \citep{TA:2008}.
However, $K$ can also be a more general function, e.g. the toroidal kernels used in \cite{DM:2011}.
\cite{TE:2022} introduced so-called delta sequence estimators and analyzed their asymptotic behavior in detail.

Smooth estimation of a circular CDF has been studied in \cite{DM:2012} for the von Mises kernel, and for general kernels in \cite{AG:2024}.

\smallskip
The structure of the paper is as follows. 
In subsection \ref{subsec11},  we review the literature on nonparametric estimators based on Bernstein polynomials and similar approximations in the linear case and show the connections to kernel estimators. Section 2 introduces the density estimator for circular distributions based on Fej\'er polynomials and gives first properties. In particular, we show the uniform strong consistency using Fej\'er's theorem. Asymptotic expansions for the Fej\'er density estimator and plug-in bandwidth selection are discussed in subsections \ref{sec:asym-exp} and \ref{sec:bandwith-select}.

Section \ref{sec:CDF-est} addresses CDF estimation based on Fejér polynomials, with asymptotic expansions for this estimator derived in \ref{sec:Fejer-CDF}. We also propose a data-dependent method for selecting the origin in the estimation of the distribution function.
Density estimation with classical and Berkson measurement error is the topic of Section \ref{sec:meas-error}. A simulation study in Section \ref{sec:sim-study} explores the finite sample properties of the Fej\'er estimators for densities both with and without measurement error, as well as the CDF.
An application to rainfall real and some remarks conclude the main part of this paper. Most of the proofs are relegated to an appendix.

\subsection{Overview over two dual approaches for nonparametric estimation of density and CDF in the linear case} \label{subsec11}

\cite{VI:1975} seems to be the first who used Bernstein polynomials $P_{k,m}(x)=\binom{m}{k}x^k(1-x)^{m-k}$ for density estimation. \cite{BA:2002, LE:2012} studied the corresponding Bernstein CDF estimator
\begin{align}
    \hat{F}_{m,n}^B(x) &= \sum_{k=0}^m F_n\left(\frac{k}{m}\right)  P_{k,m}(x)  \label{Bernstein-CDF1} \\
        &= \frac{1}{n}\sum_{i=1}^n P(Y \geq \lceil mX_i\rceil), \quad x\in [0,1], \nonumber
\end{align}
where $Y$ is a $\text{Bin}(m,x)$-distributed random variable. Noting that $P(Y \geq l)=P(Z\leq x)$ for a random variable $Z$ following a $\text{Beta}(l,m-l+1)$ distribution, we can write this as
\begin{align}
    \hat{F}_{m,n}^B(x) &= \frac{1}{n}\sum_{i=1}^n F(x; \lceil mX_i\rceil, m-\lceil mX_i\rceil+1), \label{Bernstein-CDF2}
\end{align}
where $F(\cdot;\alpha,\beta)$ denotes the CDF of a beta distribution with parameters $\alpha,\beta>0$.
Hence, these estimators use (integrated) beta kernels with expected value close to $X_i$ and variance decreasing to zero if the bandwidth $1/m$ goes to zero, similar to kernel density or CDF estimation. However, the kernels are not symmetric, have support $[0,1]$, and their shape changes according to the value of $X_i$ (but is independent of $x$).
Since the rationale of the Bernstein CDF estimator lies in Bernstein's proof of the Weierstrass approximation theorem, one can expect good performance. In particular, one obtains directly the almost certain uniform convergence of $\hat{F}_{m,n}^B$ to $F$ \citep[Th. 2.1]{BA:2002}.
By differentiating (\ref{Bernstein-CDF2}) with respect to $x$, one obtains the Bernstein density estimator
\begin{align}
    \hat{f}_{m,n}^B(x) &= \frac{1}{n}\sum_{i=1}^n f(x; \lceil mX_i\rceil, m-\lceil mX_i\rceil+1), \label{Bernstein-df}
\end{align}
where $f(\cdot;\alpha,\beta)$ denotes the density of a beta distribution with parameters $\alpha,\beta$. This estimator coincides with the estimator $\tilde{f}_{m,n}$ considered by  \cite{BA:2002} and many others.
\cite{CH:1999} directly used beta kernels for density estimation. However, his approach is different because the kernel shape changes according to the value of $x$ which changes the amount of smoothing applied by the estimators.

Bernstein density and CDF estimators are designed for functions on compact intervals. For positive random variables, the Bernstein polynomials can be replaced by generalized ones using Poisson weights. This is motivated by a uniform approximation property of the corresponding polynomials proved by \cite{SZ:1950}. This method was implemented by \cite{GS:1980} for density estimation and by \cite{HK:2021} for CDF estimation, the latter leading to
\begin{align}
    \hat{F}_{m,n}^S(x) &=\sum_{k=0}^{\infty}F_n\left(\frac{k}{m}\right)\cdot e^{-mx}\frac{(mx)^k}{k!} \nonumber\\
        &=\frac{1}{n}\sum_{i=1}^n P(Y \geq \lceil mX_i\rceil) \nonumber\\
        &= \frac{1}{n}\sum_{i=1}^n P(Z\leq x), \label{Szasz-CDF}
\end{align}
with $Y \sim Poi(mx)$ and  $Z\sim \text{Gamma}(\lceil mX_i\rceil),m)$. 
Again, this approach yields directly the almost certain uniform convergence of $\hat{F}_{m,n}^S$ to $F$ \citep[Th. 3]{HK:2021}.
By differentiating (\ref{Szasz-CDF}), one obtains the corresponding density estimator
\begin{align}
    \hat{f}_{m,n}^S(x) &= \frac{1}{n}\sum_{i=1}^n g(x; \lceil mX_i\rceil, m), \label{Szasz-df}
\end{align}
where $g(\cdot;\alpha,\beta)$ denotes the density of a gamma distribution with shape parameters $\alpha$ and rate $\beta$. This estimator coincides with the estimator considered by  \cite{GS:1980}.
\cite{CH:1999} directly considered density estimation with gamma kernels, applying the kernel $g(X_i; x/h+1, 1/h)$. Again, his approach is different since the kernel shape changes according to the value of $x$. 

\cite{CS:1996} considered smooth estimation of survival and density function. The corresponding CDF estimator is given by 
\begin{align}
\tilde{F}_{n}(x) &= c(n,\lambda_n) \sum_{k=0}^{n} F_n(k/\lambda_n)(\lambda_n x)^k/k!, \label{CS-CDF}\end{align} 
consisting only of a finite sum with weights normalized to one. 
In \cite{CS:1998} the method has been extended to density, survival, hazard and cumulative hazard function estimation for randomly censored data by replacing the edf with the Kaplan-Meier estimator. \cite{CS:1999} proposed a smooth estimator for the mean residual life function in the uncensored case. They show that $\tilde{F}_{n}$ in (\ref{CS-CDF}) is not appropriate for smooth estimation of mean residual life, and give a modification which corresponds to the Szasz estimator above. 
In \cite{CS:2008}, they combine this approach with the Kaplan-Meier estimator yielding a smooth estimator of the mean residual life function based on randomly censored data. 
There also exist several proposals for estimating density and CDF of positive random variables by other asymmetric kernels: \cite{SC:2004} used inverse and reciprocal inverse Gaussian kernels for density estimation, whereas \cite{JK:2003} applied kernels from Birnbaum-Saunders and lognormal distributions.

\section{Circular density estimation based on Fej\'er polynomials} \label{sec:dens-est}

As argued above, the counterparts of Bernstein polynomials for periodic functions are Fej\'er polynomials, and Bernstein's proof of Weierstrass' approximation theorem has Fej\'er's theorem as its counterpart:

\begin{theorem}[Fej\'er, 1915] \label{Fejer-theorem}
Let $f:\R\to\C$ be a continuous function with period $2\pi$. Then, as $m \rightarrow \infty$,
\begin{align*}
  \sigma_m(f;\theta) &= \frac{1}{2\pi} \sum_{k=-m}^m \left(1-\frac{|k|}{m+1}\right) \phi_k \, e^{ik\theta} 
  \ \rightarrow \ f(\theta)
\end{align*}
uniformly for $\theta \in [-\pi,\pi]$, where 
\begin{align*}
   \phi_k &= \int_{-\pi}^{\pi} f(t) \, e^{-ikt} \, dt, \ k\in\Z,
\end{align*}
are the Fourier coefficients.
\end{theorem}

Let 
\begin{align*}
K_m(s) &= \frac{1}{2\pi} \sum_{k=-m}^m \Big(1-\frac{|k|}{m+1}\Big) e^{iks} 
= \frac{1}{2\pi} + \frac{1}{\pi}\sum_{k=1}^m \Big(1-\frac{k}{m+1}\Big) \cos(ks)
\end{align*}
denote the \emph{Fej\'er kernel of order} $m$. Then,
\begin{align*}
\sigma_m(f;\theta) &= \int_{-\pi}^{\pi} f(t) \, K_m(\theta-t) \, dt 
 = \int_{-\pi}^{\pi} f(\theta-y) \, K_m(y) \, dy.    
\end{align*}

\begin{remark} \label{rem-fejer}
The following properties of Fej\'er kernels are well-known.
\begin{enumerate}
\item It holds that $K_m(0) = (m+1)/(2\pi)$, and 
\begin{align*}
K_m(s) &= \frac{1}{2\pi(m+1)} \left( \frac{\sin\frac{(m+1)s}{2} }{ \sin\frac{s}{2} } \right)^2, 
\; s\neq 0.    
\end{align*}
\item $K_m$ is continuous, symmetric and periodic.
\item $K_m(s)\geq 0$ for all $s$.
\item $\lim_{m\to\infty} \sup_{|s|>\delta} K_m(s)=0$, for all $\delta > 0$.
\item $\int_{-\pi}^{\pi} K_m(s)\,ds=1$ for all $m$.
\item $\int_{-\pi}^{\pi} K_m^2(s)\,ds< \infty$ for all $m$.
\end{enumerate}
\end{remark}

\smallskip
Now, let $X_1, X_2,...$ be a sequence of i.i.d random variables from an (unknown) distribution on the circle with CDF $F$ and density $f$. In analogy to the Bernstein density estimator, a circular density estimator based on a random sample $X_1,...,X_n$ can be based on  Fej\'er's approximation result, just replacing the unknown Fourier coefficients $\phi_k$ by its empirical counterparts
\begin{align*}
   \hat{\phi}_{n,k} &= \frac{1}{n} \sum_{j=1}^n e^{-ikX_j}, \ k\in\Z,
\end{align*}
corresponding to the empirical characteristic function, evaluated at point $k$. Then, the estimator is of the form
\begin{align} \label{dens-est}
  \hat{f}_{m,n}(\theta) &= \frac{1}{2\pi} \sum_{k=-m}^m \left(1-\frac{|k|}{m+1}\right) \hat{\phi}_{n,k} \, e^{ik\theta}. 
\end{align}
Alternatively, it can be written as
\begin{align} 
  \hat{f}_{m,n}(\theta)  &= \frac{1}{2\pi n} \sum_{j=1}^n\sum_{k=-m}^m \left(1-\frac{|k|}{m+1}\right) e^{ik(\theta-X_j)} \nonumber \\
  &= \frac{1}{n} \sum_{j=1}^n K_m(\theta-X_j).  \label{dens-est2}
\end{align}

Clearly, by the properties of $K_m$ in Remark \ref{rem-fejer}, $\hat{f}_{m,n}$ is a proper circular density for each $m\in \N$, i.e. $\hat{f}_{m,n}$ is a non-negative continuous periodic function with $\int_{-\pi}^{\pi}\hat{f}_{m,n}(s)ds=1$.

Our first result shows the uniform strong consistency of $\hat{f}_{m,n}$. 

\begin{theorem} \label{consitency-f}
  Assume that $m=m(n)=n^\gamma$ for some $\gamma\in (0,1/2)$.
  If $f$ is a continuous periodic circular density on $[-\pi,\pi)$, then
  \begin{equation*}
    \left\|\hat{f}_{m,n}-f\right\|\rightarrow 0 \enspace \text{a.s.}
  \end{equation*}
  for $n \rightarrow \infty$, where $\|g\|=\sup_{x\in[-\pi,\pi]}|g(x)|$ for a bounded function $g$ on $[-\pi,\pi]$.
\end{theorem}

\begin{proof}
Writing $\sigma_m(f)=\sigma_m(f;\cdot)$, we have
\begin{equation*}
  \left\|\hat{f}_{m,n}-f\right\| \leq \left\|\hat{f}_{m,n}-\sigma_m(f)\right\| + \|\sigma_m(f)-f\|.
\end{equation*}
Clearly, $\sigma_m(f)\to f$ uniformly in $\theta$ by Theorem \ref{Fejer-theorem}.
For the first summand, note that 
\begin{align*}
  \left|\hat{f}_{m,n}(\theta) - \sigma_m(f; \theta)\right| 
    &= \left| \frac{1}{2\pi} \sum_{k=-m}^m \left(1-\frac{|k|}{m+1}\right) \left( \hat{\phi}_{n,k} - \phi_k \right) \, e^{ik\theta} \right| \\
    &\leq m \cdot \, \sup_{|k|\leq m}  \left| \hat{\phi}_{n,k} - \phi_k \right|.
\end{align*}
By Example 1 in \cite{CS:1985}, we obtain for any $\gamma>0$
\begin{align*}
\limsup_{n\to\infty} \sqrt{\frac{n}{\log n}} \ \sup_{|t|<n^\gamma} \left|\psi_n(t)-\psi(t)\right| \leq c, \quad a.s.,    
\end{align*}
where $c$ is a positive constant, and $\psi_n$ and $\psi$ denote the empirical c.f. and c.f. of a distribution with a tail decreasing at least like a power function. 
Since $\hat{\phi}_{n,k}$ and $\phi_k$ correspond to $\psi_n$ and $\psi$, evaluated at the integers, from a distribution with compact support $[-\pi,\pi]$, this result 
implies for $m=n^\gamma$ that 
\begin{align*}
  \left\| \hat{f}_{m,n} - \sigma_m(f)\right\| &= O\left( n^{\gamma-1/2} \left(\log n\right)^{1/2} \right) = o(1).
\end{align*}
Therefore, the assertion follows. 
\end{proof}

Figure \ref{fig2} shows the Fej\'er estimator for random samples of size $n=50$ and $400$, respectively, from the wrapped Cauchy distribution with $\mu=\pi/2,\, \rho=\exp(-1)$, and different values of the order $m$, which takes on the role of the (inverse) bandwidth in kernel density estimation.

\begin{figure}
  \centering
 \includegraphics[width=\textwidth]{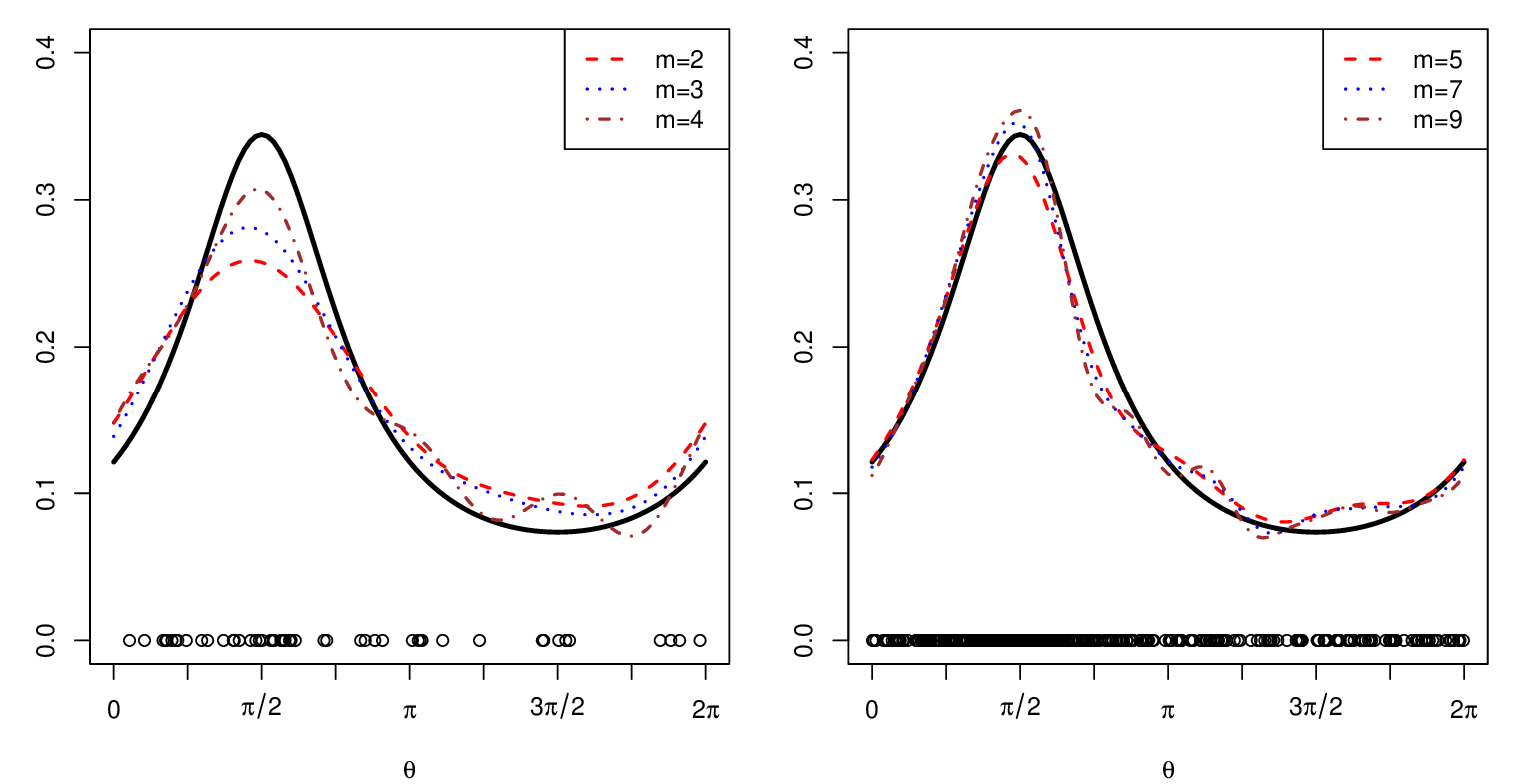}
\caption{Fej\'er density estimator for samples of size $50$ (left) and $400$ (right) from a wrapped Cauchy distribution with different values of the bandwidth $m$. \label{fig2}}
\end{figure}
\subsection{Asymptotic expansions for the Fej\'er density estimator} \label{sec:asym-exp}

\cite{TE:2022} introduced the concept of a delta sequence estimator in the circular case, that is an estimator of the form $\hat{f}_n(\theta)=1/n \sum_{i=1}^n \delta_n(\theta-X_i)$, where $\delta_n:\R\to [0,\infty)$ is a sequence of $2\pi$-periodic functions satisfying three conditions $(\Delta.1)-(\Delta.3)$ which correspond to properties (iv)-(vi) in Remark \ref{rem-fejer}. Hence, assuming that $m=m(n)\to \infty$ when $n\to \infty$, the Fej\'er density estimator belongs to this class.
By the symmetry of $K_m$, condition $(\Delta.4)$ in \cite[p. 382]{TE:2022} is also fulfilled.
Further, we have 
\begin{align} \label{alpha}
    \alpha(K_m) &= \int_{-\pi}^{\pi} K_m^2(y)dy  \nonumber \\
    &= \int_{-\pi}^{\pi} \left( \frac{1}{\pi^2} + \frac{1}{4\pi^2} \sum_{k=1}^m \left(1-\frac{k}{m+1}\right)^2 \cos^2(ky) \right) dy \nonumber \\
    &= \frac{1}{2\pi} + \frac{1}{\pi} \sum_{k=1}^m \left(1-\frac{k}{m+1}\right)^2 
    \,=\, \frac{1}{2\pi} + \frac{m(2m+1)}{6\pi(m+1)} 
    \,\sim \, \frac{m}{3\pi}.
\end{align}
Assume that $m\to\infty$ and $n/m\to \infty$, as $n\to\infty$. Then, $\alpha(K_m)\to\infty$ and $n\alpha(K_m)^{-1} \to\infty$, as $n\to\infty$. Consequently, Theorem 2.1 and formula (12) in \cite{TE:2022} hold. Utilizing (\ref{alpha}), we obtain for a continuous density $f$ on $[-\pi,\pi]$
\begin{align*}
\sup_{\theta\in[-\pi,\pi]}& \left| E\hat{f}_{m,n}(\theta)-f(\theta) \right| \to 0,  \\
\sup_{\theta\in[-\pi,\pi]}& \left| \frac{3\pi n}{m} Var\left(\hat{f}_{m,n}(\theta)\right) - f(\theta) \right| \to 0,  
\end{align*}
as $n\to\infty$. It follows that 
\begin{align*}
\sup_{\theta\in[-\pi,\pi]} \left| Var\left(\hat{f}_{m,n}(\theta)\right) - \frac{m}{3\pi n}  f(\theta) \right| = o\left(\frac{m}{n}\right), \quad  \text{as } n\to\infty. 
\end{align*}
Then, the integrated variance can be expressed as 
\begin{align*}
IV(f;\hat{f}_{m,n}) &= \int_{-\pi}^{\pi} Var\left(\hat{f}_{m,n}(\theta)\right) d\theta 
= \frac{m}{3\pi n} + o\left(\frac{m}{n}\right), \quad  \text{as } n\to\infty. 
\end{align*}
Next, we want to derive an expansion for the integrated squared bias. Define
\begin{align*}
\beta(K_m) &= \int_{-\pi}^{\pi} y^2 K_m(y)dy, \quad 
\gamma(K_m) = \int_{-\pi}^{\pi} |y|^3 K_m(y) dy.
\end{align*} 
Proposition \ref{prop-moments} in Appendix \ref{app:A} gives the exact asymptotic behavior of $\beta(K_m)$ and $\gamma(K_m)$. It follows that $\beta(K_m)/\gamma(K_m)\to c >0$ when $m\to\infty$. 
Thus, the condition $(\Delta.5)$ for delta sequence estimators assumed in \cite{TE:2022} is not satisfied by $K_m$, and Theorems 2.2 and 3.1 in \cite{TE:2022} cannot be applied. Similarly, the approximation of the mean integrated squared error in \citet[Theorem 1]{DM:2009} is invalid for the Fej\'er kernel, even though this kernel is mentioned as an example in that paper. The reason is as follows: the assumptions and results of \citet{DM:2009} are collected in Result 1 of \citet{DM:2022}. For the Fej\'er kernel, with $\gamma_k(m)=\ind_{k\leq m} \left(1-k/(m+1)\right)$, 
\begin{align*}
    \lim_{m \to \infty} \frac{1-\gamma_k(m)}{1-\gamma_2(m)} &= \frac{k}{2}.
\end{align*}
Therefore, assumption (ii) in Result 1 is not satisfied, and the asymptotic expansion for the bias in Theorem 1 of \citet{DM:2009} does not apply to the Fej\'er kernel.

Instead, we can proceed as \cite{TE:2022} in the case of the wrapped Cauchy kernel. The crucial point is the representation of the integrated squared bias as 
\begin{align*}
ISB(f;\hat{f}_{m,n}) &= \int_{-\pi}^{\pi} \left(E\hat{f}_{m,n}(\theta) - f(\theta) \right)^2 d\theta \\
&= \frac{1}{\pi} \sum_{k=1}^\infty \left(1-\gamma_k(m)\right)^2 \left(a_k^2+b_k^2\right)
\end{align*}
for a square-integrable function $f$ in $[-\pi,\pi]$, where the trigonometric
moments $a_k$ and $b_k$ are given by 
$a_k=\int_{-\pi}^{\pi} f(\theta)\cos(k\theta)d\theta$ and $b_k=\int_{-\pi}^{\pi} f(\theta)\sin(k\theta)d\theta$. For a function $f$ that is absolutely continuous on $[-\pi,\pi]$ and has a square-integrable derivative $f'$ on $[-\pi,\pi]$, we obtain  
\begin{align}
ISB(f;\hat{f}_{m,n}) \nonumber
&= \frac{1}{\pi(m+1)^2} \sum_{k=1}^m k^2 \left(a_k^2+b_k^2\right)
 + \frac{1}{\pi} \sum_{k=m+1}^\infty \left(a_k^2+b_k^2\right)\\\nonumber
&= \frac{\theta_1(f)}{(m+1)^2} + \frac{1}{\pi} \sum_{k=m+1}^\infty \left( 1-\frac{k^2}{(m+1)^2} \right) \left(a_k^2+b_k^2\right) \\
&= \theta_1(f)/(m+1)^2 + o((m+1)^{-2}),\label{ISB}
\end{align}
where 
\begin{align} \label{theta1}
\theta_1(f) & = \frac{1}{\pi} \sum_{k=1}^\infty k^2 \left(a_k^2+b_k^2\right) 
= \int_{-\pi}^{\pi} f'(\theta)^2 d\theta
\end{align}
(see the proof of Theorem 3.2 in \cite{TE:2022}). Summarizing, we have the following result.

\begin{theorem}
Let $f$ be continuously differentiable on $[-\pi,\pi]$. If $m\to\infty$ and $n/m\to \infty$, as $n\to\infty$, then 
\begin{align*}
\MISE(f;\hat{f}_{m,n}) &= \frac{m}{3\pi n} + \frac{\theta_1(f)}{m^2} + o\left(\frac{m}{n}+ \frac{1}{m^2}\right),      
\end{align*}
where $\theta_1(f)= \int_{-\pi}^{\pi} f'(\theta)^2 d\theta$.
\end{theorem}

If $f$ is not the circular uniform distribution, the asymptotic optimal choice of $m$ is given by 
\begin{align} \label{optimal-m}
m_{opt} &= (6\pi\theta_1(f))^{1/3} \, n^{1/3},
\end{align}
leading to the asymptotic mean integrated squared error
\begin{align*}
\AMISE(f;\hat{f}_{m_{opt},n}) &= c_F \left( \frac{\theta_1(f)}{\pi^2} \right)^{1/3} n^{-2/3},
\quad c_F = \left(\frac{2}{9}\right)^{1/3} + \left(\frac{1}{6}\right)^{2/3}.
\end{align*}
The AMISE for the estimator based on the wrapped Cauchy kernel given in \cite{TE:2022} has the same form, with the constant $c_F\approx 0.91$ replaced by the slightly larger constant 
$c_{WC}=2^{-1/3} + 4^{-2/3}\approx 1.19$.

\subsection{Plug-in bandwidth selection} \label{sec:bandwith-select}

The only unknown quantity in the asymptotic optimal choice of $m$ in (\ref{optimal-m}) is  $\theta_1(f)$. The simplest approach assumes that $f$ is a member of some parametric family of distributions, such as the von Mises family. After estimating the unknown parameters, direct computation of  $\theta_1(f)$ is possible. This does not lead to the optimal choice of $m$ if the data does not come from the assumed distribution.
A more general way for estimating $\theta_1(f)$ uses the representation in (\ref{theta1}). One possible approach estimates $a_k$ and $b_k$ by
\begin{align*}
    \hat{a}_k&= \frac{1}{n} \sum_{i=1}^n \cos(kX_i) \quad \text{and} \quad 
    \hat{b}_k= \frac{1}{n} \sum_{i=1}^n \sin(kX_i),
\end{align*}
leading to 
\begin{align} \label{plugin-first-derivative}
\tilde{\theta}_{1,M} & = \frac{1}{\pi} \sum_{k=1}^M k^2 \left(\hat{a}_k^2+\hat{b}_k^2\right),
\end{align}
where $M=M(n)$ converges to infinity.
Alternatively, one can use $\hat{\theta}_{1,M}=1/\pi \sum_{k=1}^M k^2 \hat{c}_k$, where $\hat{c}_k$ is the unbiased estimator of $a_k^2+b_k^2$ given by \citet[p. 390]{TE:2022} 
\begin{align*} 
\hat{c}_k & = \frac{2}{n(n-1)} \sum_{1\leq i,j\leq n} \cos\left(k(X_i-X_j)\right).
\end{align*}
In practice, the choice of $M$ should be data-dependent, i.e. $M=\hat{M}(X_1,\ldots,X_n)$. From Lemma 1 in \citet{TE:2011}, one can deduce that if $\hat{M}$ is such that $\hat{M}\cp\infty$ and $n^{-1}\hat{M}^3\cp 0$, then $\hat{\theta}_{1,\hat{M}}\cp\theta_1$.
Summarizing, we obtain
\begin{align*}
    \frac{\big(6\pi\hat{\theta}_{1,\hat{M}}\big)^{1/3} \, n^{1/3}}{m_{opt}} &\cp 1.
\end{align*}

\section{Distribution function estimation based on Fej\'er polynomials}\label{sec:CDF-est}

Associated with the circular kernel $K_m$ is the integrated kernel 
\begin{align} \label{int-kernel}
    W_m(\theta) &= \int_{-\pi}^{\theta} K_m(y)dy 
    = \frac{\theta+\pi}{2\pi} + \frac{1}{\pi} \sum_{k=1}^m \left(1-\frac{k}{m+1}\right) \frac{1}{k} \sin(k\theta).
\end{align}
A possible estimator of a circular CDF $F(\theta)=\int_{-\pi}^{\theta} f(y)dy$ is obtained by integrating over the density estimator in (\ref{dens-est}), leading to 
\begin{align*}
  \hat{F}_{m,n}(\theta) &= \int_{-\pi}^{\theta} \hat{f}_{m,n}(y)dy \\
 &= \frac{1}{n} \sum_{i=1}^n \left\{ W_m(\theta-X_i) - W_m(-\pi-X_i) \right\}.
\end{align*} 
Note that the estimator depends on the lower limit of the integration $\theta_0$. Although it seems natural to use $\theta_0=-\pi$, the discussion in \cite{DM:2012} shows that this is not an optimal choice, in general. We will discuss this in more detail in Section \ref{sec:low-lim}.
Up to this point, we always use $\theta_0=-\pi$.

\smallskip
The next theorem shows that $\hat{F}_{m,n}$ is uniformly strongly consistent as long as $m$ does not increase faster than a power of $n$. 

\begin{theorem}
  Assume that $m=m(n)=n^\gamma$ for some $\gamma>0$.
  If $F$ is a continuous distribution function on $[-\pi,\pi)$, then
  \begin{equation*}
    \left\|\hat{F}_{m,n}-F\right\|\rightarrow 0 \enspace \text{a.s.}
  \end{equation*}
  for $n \rightarrow \infty$.
\end{theorem}

\begin{proof}
First,
\begin{equation*}
  \left\|\hat{F}_{m,n}-F\right\| \leq \left\|\hat{F}_{m,n}-\sigma_m(F)\right\| + \|\sigma_m(F)-F\|.
\end{equation*}
Since $\sigma_m(f)\to f$ uniformly in $\theta$ by Theorem \ref{Fejer-theorem}, it follows that \
\begin{align*}
    |\sigma_m(F;\theta)-F(\theta)| &= \left| \int_{-\pi}^{\theta} \left( \sigma_m(f;y)-f(y) \right) dy \right| \leq 2\pi\epsilon
\end{align*}
for $\epsilon>0$ and $m$ large enough; hence $\|\sigma_m(F)-F\|\to 0$ as $m\to \infty$.
For the first summand, note that 
\begin{align*}
  \left|\hat{F}_{m,n}(\theta) - \sigma_m(F; \theta)\right| 
    &= \left| \int_{-\pi}^{\theta} \frac{1}{2\pi} \sum_{k=-m}^m \left(1-\frac{|k|}{m+1}\right) \left( \hat{\phi}_{n,k} - \phi_k \right) e^{iky} dy \right| \\
    &\leq \frac{1}{\pi} \sum_{k=1}^m \left(1-\frac{k}{m+1}\right) \frac{2}{k} \, \cdot \, \sup_{|k|\leq m}  \left| \hat{\phi}_{n,k} - \phi_k \right| \\
    &\leq H_m \, \sup_{|k|\leq m} \left| \hat{\phi}_{n,k} - \phi_k \right|,
  \end{align*}
where $H_m$ is defined in Lemma \ref{lemma-harmonic}. 
Here, $H_m \sim \log m=\gamma \log n$ for $m=n^\gamma$. 
Again using Example 1 in \cite{CS:1985} as in the proof of Theorem \ref{consitency-f}, we obtain $|\hat{F}_{m,n}-\sigma_m(F)|\to 0$ as $n\to \infty$ and $m=n^{\gamma}$, for any $\gamma>0$.
\end{proof}

\subsection{Asymptotic expansions for the Fej\'er CDF estimator} \label{sec:Fejer-CDF}

Defining formally the concentration parameter $\rho=1-1/m\in[0,1)$, the Fej\'er Kernels $K_m$ fulfils assumptions (K1), (K2), (K3) and (K10) in \cite{AG:2024}.
Assuming again that $m\to \infty$ and $m/n\to 0$, as $n\to\infty$, condition (P1) and, hence, conditions
(C1)–(C3) are fulfilled.  

Next, define $\nu_{1,m} = 2\pi \int_{-\pi}^{\pi} y \cdot W_m(y) \cdot K_m(y) \, dy,$
which corresponds to $m_{1;s}(\rho_s)$ in equation (7) of \cite{AG:2024}.
Lemma \ref{lemma-m1} in Appendix \ref{app:B} gives the asymptotic behavior of $\nu_{1,m}$, showing that
\[
\nu_{1,m} = \frac{2 \log m}{m+1} + O(m^{-1}).
\]
Further, define 
\begin{align*}
\gamma_{1,m}(\theta) &= \theta + 2 \gamma_{m}(\theta), \\
\gamma_{2,m}(\theta) &= \frac{1}{2} + \frac{1}{\pi} \left( -\sum_{k=1}^{m} \left(\frac{1}{k}-\frac{1}{m+1}\right) \sin(k\theta) + \gamma_{m}(\theta) \right),
\end{align*}
where
\[
\gamma_{m}(\theta) = \sum_{k=1}^{m} (-1)^k \left(1-\frac{k}{m+1}\right)^2 \frac{1}{k} \sin(k\theta).
\]
Then, by Lemma \ref{lemma-gamma1-2},
\[
\gamma_{1,m}(\theta) = O(m^{-1})=o(\nu_{1,m}) \quad \text{and} \quad  \gamma_{2,m}(\theta) = O(m^{-1}) = o(\nu_{1,m}),
\]
as $m\to\infty$. Accordingly, conditions (P2) and, hence, (K5) in \cite{AG:2024} are fulfilled.
Finally, Lemma \ref{lemma-m3} shows that 
\[
\nu_{3,m} = \int_{-\pi}^{\pi} y^3 \cdot W_m(y) \cdot K_m(y) \, dy = O(m^{-1}) = o(\nu_{1,m}).
\]
Thus, the condition (K7) in \cite{AG:2024} is also satisfied. Under the assumption
\[
\text{(F) \quad The function $F$ is twice continuously differentiable on $[-\pi,\pi)$},
\]
their Th. 2.2(v) yields an explicit expression for the asymptotic variance:
\begin{align*}
\AVar\left(\hat{F}_{m,n}(\theta)\right) 
&=\frac{F(\theta)(1-F(\theta))}{n} - \frac{1}{\pi n} \nu_{1,m} \left( F'(\theta) + F'(-\pi) \right). 
\end{align*}
Again using Lemma \ref{lemma-m1}, we obtain
\begin{align} \label{AVar-F}
\AVar\left(\hat{F}_{m,n}(\theta)\right) 
&=\frac{F(\theta)(1-F(\theta))}{n} - \frac{2\log m}{\pi mn} \left( F'(\theta) + F'(-\pi) \right). 
\end{align}
We will now focus on the integrated squared bias 
\[
\ISB(F;\hat{F}_{m,n}) = \int_{-\pi}^{\pi} \left( E\hat{F}_{m,n} - F(\theta) \right)^2 d\theta.
\]
By Lemma \ref{prop-high-mom} in the appendix, 
\begin{align*}
m_4(K_m) &= \int_{-\pi}^{\pi} y^4 K_m(y)dy \ = \frac{8\pi^2\log(2)-36\zeta(3)}{m} + O(m^{-2}).
\end{align*} 
It follows that assumption (K6) in \citet{AG:2024} is not satisfied, and the explicit expression for the asymptotic bias in their Theorem 2.2(v), corresponding to the asymptotic bias derived in \cite{DM:2012} for the von Mises kernel, is not applicable. 
Again, there is a great similarity between the Fej\'er and the wrapped Cauchy kernel, which also does not fulfil (K6). Instead, we have the following result.

\begin{theorem} \label{F-bias}
Under assumption (F), the integrated squared bias of $\hat{F}_{m,n}$ is given by
\begin{align*}
\ISB(F;\hat{F}_{m,n}) &= \frac{\theta_2(F)}{(m+1)^2} + o(m^{-2}) ,
\end{align*}
where 
\begin{align} \label{theta2-F}
\theta_2(F) &= \frac{1}{\pi}\sum_{k=1}^{\infty} \left( a_k^2+b_k^2 \right) 
+ \frac{2}{\pi}\sum_{k,l=1}^{\infty} (-1)^{k+l} b_k b_l 
= \int_{-\pi}^{\pi} \left( f^{\sim}(\theta) - f^{\sim}(-\pi) \right)^2 d\theta,
\end{align}
and $f^{\sim}(\theta)=1/\pi \sum_{k=1}^{\infty} \left(a_k \sin(k\theta) - b_k \cos(k\theta) \right)$ denotes the Hilbert transform of $f=F'$.
\end{theorem}

\begin{proof}
For a general integrated kernel, write
\[ W(y) = \frac{y + \pi}{2\pi} + \frac{1}{\pi} \sum_{k=1}^{\infty} \frac{\gamma_k}{k} \sin(ky). \]
First, we compute 
\begin{align*}
E\hat{F}_{m,n}(\theta) 
&= \int_{-\pi}^{\pi} \left( W(\theta - y) - W(-\pi - y) \right) f(y) \, dy.
\end{align*}
Using $\sin(k(\theta - y)) = \sin(k\theta) \cos(ky) - \cos(k\theta) \sin(ky)$, we obtain
\begin{align*}
W(\theta - y) & - W(-\pi - y) \\
& = \frac{\theta + \pi}{2\pi} + \frac{1}{\pi} \sum_{k=1}^{\infty} \frac{\gamma_k}{k} (\sin(k\theta) \cos(ky) - \cos(k\theta) \sin(ky) + (-1)^k \sin(ky)).
\end{align*}
Hence,
\begin{align*}
&E\hat{F}_{m,n}(\theta) = \int_{-\pi}^{\pi}  \left( \frac{\theta + \pi}{2\pi} + \frac{1}{\pi} \sum_{k=1}^{\infty} \frac{\gamma_k}{k} (\sin(k\theta) \cos(ky) - \cos(k\theta) \sin(ky) + (-1)^k \sin(ky)) \right) \\
&\hspace*{24mm} \left( \frac{1}{2\pi} + \frac{1}{\pi} \sum_{k=1}^{\infty} \left(a_k \cos(ky) + b_k \sin(ky) \right) \right) dy \\
&=  \frac{\theta+\pi}{2\pi} + \frac{1}{\pi^2} \sum_{k=1}^{\infty} \frac{\gamma_k}{k}
\int_{-\pi}^{\pi} \left( a_k\sin(k\theta) \cos^2(ky) - b_k\cos(k\theta)\sin^2(ky) + (-1)^k b_k\sin^2(ky) \right) \\
&= \frac{\theta + \pi}{2\pi} +  \frac{1}{\pi} \sum_{k=1}^{\infty}  \frac{\gamma_k}{k} \left( a_k \sin(k\theta) - b_k \cos(k\theta) + (-1)^k b_k \right).
\end{align*}
To compute the integrated squared bias, note that 
\begin{align*}
F(\theta) 
&=\frac{\theta+\pi}{2\pi} + \frac{1}{\pi} \sum_{k=1}^{\infty} \left( \frac{a_k}{k} \sin(k\theta) - \frac{b_k}{k} \cos(k\theta) + \frac{b_k}{k} (-1)^k \right), 
\end{align*}
and thus
\begin{align*}
E\hat{F}_{m,n}(\theta) - F(\theta) 
&= \frac{1}{\pi} \sum_{k=1}^{\infty} \frac{\gamma_k-1}{k} \, \delta_k(\theta), 
\end{align*}
where 
\[
\delta_k(\theta) = a_k \sin(k\theta) - b_k \cos(k\theta) + b_k (-1)^k.
\]
Inserting the Fej\'er kernel $\gamma_k=(1-k/(m+1)), k\leq m$, yields
\begin{align*}
E\hat{F}_{m,n}(\theta) - F(\theta) 
&= \frac{-1}{(m+1)\pi} \left( \sum_{k=1}^{\infty} \delta_k(\theta) - \sum_{k=m+1}^{\infty} \frac{k-m-1}{k} \delta_k(\theta) \right) \\
&= \frac{-1}{(m+1)\pi} \sum_{k=1}^{\infty} \delta_k(\theta) + o(m^{-1}),
\end{align*}
since $\sum_{k=1}^{\infty} |\delta_k(\theta)|<\infty$ under assumption (F).
Next, we have
\begin{align*}
\int_{-\pi}^{\pi} \left(\frac{1}{\pi}\sum_{k=1}^{\infty} \delta_k(\theta) \right)^2 d\theta
&= \frac{1}{\pi}\sum_{k=1}^{\infty} \left( a_k^2+b_k^2 \right) 
+ \frac{2}{\pi}\sum_{k,l=1}^{\infty} (-1)^{k+l} b_k b_l
= \theta_2(F),
\end{align*} 
say. Since the Hilbert transform of $\cos(y)$ is $H(\cos)(y)=\sin(y)$, and $H(\sin)(y)=-\cos(y)$, we obtain, using the linearity of the transform,
\begin{align*}
f^{\sim}(\theta) & = H(f)(\theta) 
=\frac{1}{\pi} \sum_{k=1}^{\infty} \left(a_k \sin(k\theta) - b_k \cos(k\theta) \right).
\end{align*}
It follows that  
\begin{align*}
\theta_2(F) &= \int_{-\pi}^{\pi} \left( f^{\sim}(\theta) - f^{\sim}(-\pi) \right)^2 d\theta.
\end{align*}
Summarizing, we obtain
\begin{align*}
\ISB(F;\hat{F}_{m,n}) 
&= \frac{\theta_2(F)}{(m+1)^2} + o(m^{-2}).
\end{align*}
\end{proof}

From Theorem \ref{F-bias} and (\ref{AVar-F}), we obtain the following result.

\begin{corollary} \label{cor-AMISE}
Under assumption (F), and if $m\to\infty$ and $n/m\to \infty$, as $n\to\infty$, we have
\begin{align*} 
\AMISE(F;\hat{F}_{m,n}) &= \frac{1}{n} \int_{-\pi}^{\pi} F(\theta)(1-F(\theta))\, d\theta + R_{m,n},
\end{align*}
where
\begin{align} \label{AMISE_F}
R_{m,n} &= - \frac{2\log m}{\pi mn} \left( 1+2\pi F^{\prime}(-\pi) \right) + \frac{\theta_2(F)}{m^2}.
\end{align}
\end{corollary}

The asymptotic optimal choice of $m$ is given as the solution of
\begin{align} \label{optimal-m2}
\left( 1+2\pi F'(-\pi) \right) m (\log m-1) &= \pi\theta_2(F) \, n.
\end{align}
Writing $c=\pi\theta_2(F) (1+2\pi F'(-\pi))^{-1}$, equation (\ref{optimal-m2}) is equivalent to 
\begin{align*} 
u e^u &= \frac{cn}{e}, \quad \text{with} \quad u=\frac{cn}{m}. 
\end{align*}
The solution of this equation is $W_0(cn/m)$, where $W_0$ denotes the principal branch of the Lambert W function. Substituting back, we obtain as the asymptotic optimal choice of $m$
\begin{align}\label{m-opt-CDF}
    m_{opt} &= \frac{cn}{W_0(cn/e)}.
\end{align}
We note that the asymptotic behavior of $W_0$ is given by $W_0(z)=\log(z)-\log(\log(z))+o(1)$.
The sign of the remainder $ R_{m,n}$ in (\ref{AMISE_F}) depends on $c$ and $n$. 
Thus, without further knowledge, it is unclear whether the fast convergence to zero shown by the Fej\'er kernel or a slower rate like for the von Mises kernel, where the rate is $O(n^{-4/3})$ \citep{AG:2024}, is preferable. 
As for the density estimators, the results for the Fej\'er distribution function estimator are quite parallel to the results for the wrapped Cauchy kernel, where the rate of the remainder term $R_n$ is $o(n^{-\beta})$ for any $\beta <2$ (see Table 2, p. 13, in \cite{AG:2024}; the exact rate is not given).

\subsection{Optimal choice of the origin} \label{sec:low-lim}

In this subsection, we discuss a possible choice of origin when defining the CDF. Write 
\begin{align*}
    F^{\theta_0}(\theta) &= \int_{-\theta_0}^{\theta} f(y)dy, \quad
    \hat{F}_{m,n}^{\theta_0}(\theta) = \int_{-\theta_0}^{\theta} \hat{f}_{m,n}(y)dy.
\end{align*}
When choosing $\theta_0$ as origin, the expression for the AMISE in Corollary \ref{cor-AMISE} becomes 
\begin{align} \label{C-crit}
\AMISE(F^{\theta_0};\hat{F}_{m,n}^{\theta_0}) &= \frac{C(F^{\theta_0})}{n} + R_{m,n}^{\theta_0},
\quad 
C(F^{\theta_0}) = \int_{\theta_0}^{\theta_0+2\pi} F^{\theta_0}(\theta) \big(1-F^{\theta_0}(\theta)\big) d\theta,
\end{align}
where $R_{m,n}^{\theta_0}$ corresponds to $R_{m,n}$ with $F'(-\pi)$ replaced by $F'(\theta_0)$, and $\theta_2(F)$ by
\begin{align*}
    \theta_2(F,\theta_0)=\frac{1}{\pi}\sum_{k=1}^{\infty} \left( a_k^2+b_k^2 \right) 
+ \frac{2}{\pi} \bigg( \sum_{k=1}^{\infty} \left(-a_k\sin(k\theta_0) + b_k\cos(k\theta_0) \right) \bigg)^2.
\end{align*}
There are three main messages from (\ref{C-crit}):
\begin{enumerate}
\item 
As in the linear case, the dominant part is $C(F^{\theta_0})/n$, which corresponds to the variance of the empirical CDF and does not depend on $m$.
\item 
Unlike in the linear case, this dominant part depends on the choice of the origin.
\item 
The smoothing effect is captured in $R_{m,n}^{\theta_0}$, which shows that the smoothing has a second-order effect on the CDF. Thus, the choice of $m$ is less important for smooth CDF estimation.
\end{enumerate}
Given this, we propose the following two-step approach: \\
First, choose the origin by minimizing an empirical counterpart of $C(F^{\theta_0})$. To do this, transform the data set to the interval $[\theta_0,\theta_0+2\pi)$, which yields the ordered values $x_{(1)}^{\theta_0} \leq \ldots \leq x_{(n)}^{\theta_0}$, and compute the empirical CDF $F_n^{\theta_0}$ with origin $\theta_0$ based on these values. Then minimize
\begin{align} \label{Cn-crit}
 C_n(\theta_0) = C(F_n^{\theta_0}) &= \sum_{i=1}^n \frac{i}{n}  \Big(1-\frac{i}{n}\Big) \left(x_{(i+1)}^{\theta_0} - x_{(i)}^{\theta_0} \right) 
\end{align}
with respect to $\theta_0\in[-\pi,\pi)$. In practice, since the function \eqref{Cn-crit} is piecewise constant, the minimum is identified over a range of values, and the circular midpoint of this range is taken as the final estimate.
In the second step, select $m$ as described in the previous section. 

To illustrate the first step, consider Figure \ref{fig:origin}, which shows $C_n(\theta_0)$ for samples of size $n=200$ from a von Mises distribution $VM(\pi/2,0.5)$ (left panel), $VM(\pi/2,2)$ (middle panel), and from a mixture $Mix(VM(0,2),VM(\pi,3),0.5)$ of two von Mises distributions with equal mixing proportions (right panel).

\begin{figure}
    \centering
    \includegraphics[width=\textwidth]{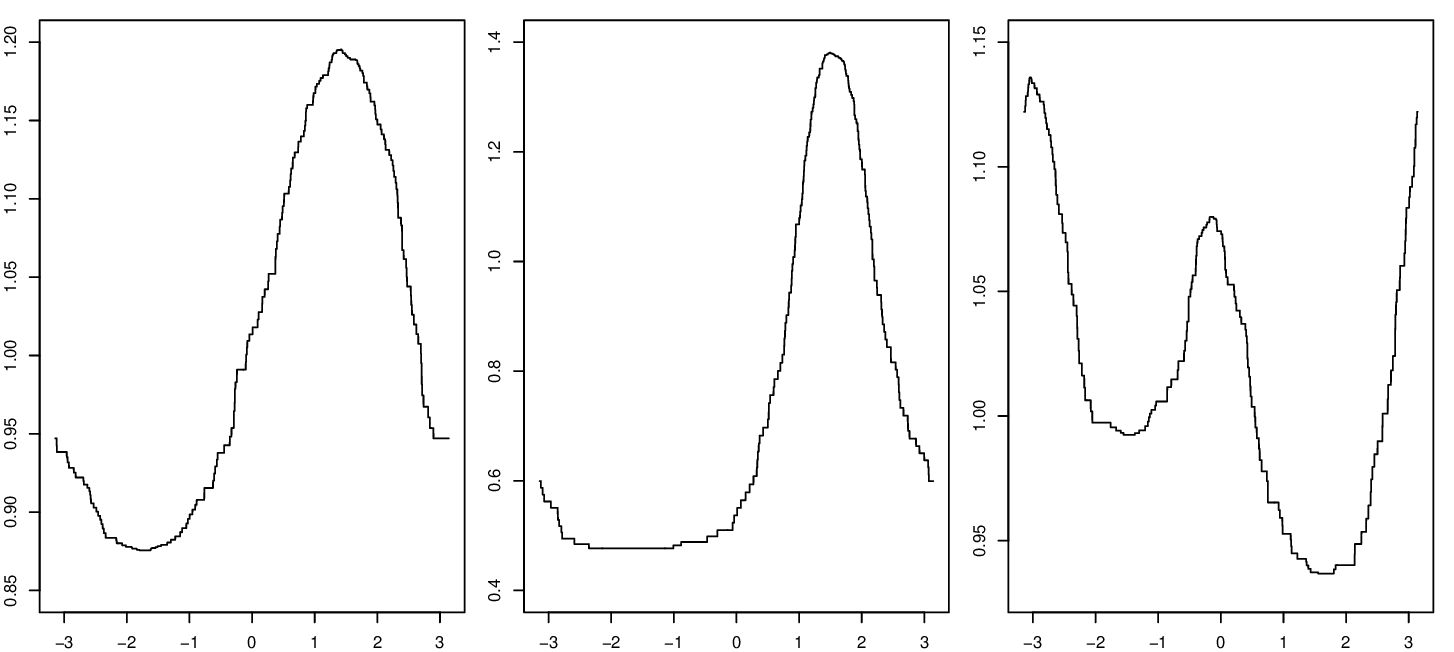}
    \caption{Plot of $C_n(\theta_0)$ in (\ref{Cn-crit}) for samples of size $n=200$ from $VM(\pi/2,0.5)$ (left panel), $VM(\pi/2,2)$ (middle panel), and $Mix(VM(0,2),VM(\pi,3),0.5)$ (right panel) for $\theta_0\in[-\pi,\pi)$. \label{fig:origin}}
\end{figure}

For $VM(\pi/2,0.5)$, the minimum and maximum values are $0.88$ and $1.20$. For the more concentrated $VM(\pi/2,2)$ distribution, we get $0.48$ and $1.38$, almost a factor of 3. For the mixture distribution, the minimum is $0.94$, quite close to the maximum of $1.08$.
These examples show that the choice of origin can greatly affect the mean integrated squared error.

\section{Density estimation with measurement error} \label{sec:meas-error}

Density estimation in the presence of measurement error is a challenging problem that has attracted the attention of many researchers.
Most of the work has been focused on linear data (see \cite{DE:2014} and references therein), while the first attempt at the case of circular data is from \cite{DM:2022}. 

Measurement error models are typically classified into two categories: Berkson and classical. Both models are additive and differ only in assuming whether the error variable is independent of the observed or unobservable data. We consider estimation under both models using the deconvolution approach in the following.

\subsection{Density estimation with Berkson error model}

We consider the model
\begin{align}\label{berksonModel}
    X=(X^{\star}+\varepsilon)\mod2\pi,
\end{align}
where the random variable $X$ with density $f_X$ is the (unobservable) quantity we want to measure, $X^\star$ with density $f_{X\star}$ is the measured value from which we have the sample $X^\star_1,\ldots,X^\star_n$. 
The random error $\varepsilon$ is independent of $X^\star$ and has a density $f_{\varepsilon}$ that is symmetric around zero.
We assume that all densities are square-integrable on $[0,2\pi)$ and allow an absolutely convergent Fourier series representation.

We assume $f_\varepsilon$ to be known with some concentration parameter $\kappa_\varepsilon$ and a Fourier representation
\begin{align*}f_\varepsilon(u)=\frac1{2\pi}\Big(1+2\sum_{j=1}^\infty\lambda_j(\kappa_\varepsilon)\cos(ju)\Big).
\end{align*}
According to \eqref{berksonModel}, for $x\in[0,2\pi)$ we have 
\begin{align}
    f_{X}(x)=\int_0^{2\pi}f_{X\star}(u)f_{\varepsilon}(x-u)du,
\end{align}
and, for $l\in\mathbb{Z}$,
\begin{align}
    \varphi_{X}(l)=\varphi_{X^\star}(l)\varphi_{\varepsilon}(l),
\end{align}
where $\varphi_{X},\varphi_{X^\star}$ and $\varphi_{\varepsilon}$ are the respective characteristic functions. Then, using the inversion formula we can obtain the simple estimator
\begin{align}\label{simpleEstimator}
    \tilde{f}_{X}(x)=\frac1{2\pi}\sum_{l=-\infty}^\infty \hat{\varphi}_{X^{\star}}(l)\varphi_{\varepsilon}(l)e^{-ilx},
\end{align}
where $\hat{\varphi}_{X^{\star}}$ is the empirical characteristic function of the sample $X^\star_1,\ldots,X^\star_n$.

Following the idea presented in \cite{SO:2021} for the analogous case of linear data, to increase smoothness we suggest a modification of \eqref{simpleEstimator} by adding a Fej\'er kernel factor
\begin{align}\label{deconvolutionEstimator}
    \hat{f}(x;m)&=\frac1{2\pi}\sum_{l=-\infty}^\infty \hat{\varphi}_{X^{\star}}(l)\varphi_{\varepsilon}(l)\varphi_{K_{m}}(l)e^{-ilx}\\
    &=\frac{1}{2\pi n}\sum_{j=1}^n\Big(1+2\sum_{l=1}^m\left(1-\frac{l}{m+1}\Big)\lambda_l(\kappa_\varepsilon)\cos(l(x-X^\star_j))\right).
\end{align}
Since the error is independent of the observed data, optimal bandwidth can be chosen as in the error-free case.

\subsection{Density estimation in the classical error model}

Here we consider the model
\begin{align}\label{classicalModel}
    X^{\star}=(X+\varepsilon)\mod2\pi,
\end{align}
where the notation is the same as in the previous section and the error variable $\varepsilon$ is independent of the unmeasurable quantity $X$.
As before, it holds for $x\in[0,2\pi)$, 
\begin{align}
    f_{X\star}(x)=\int_0^{2\pi}f_{X}(u)f_{\varepsilon}(x-u)du,
\end{align}
and, for $l\in\mathbb{Z}$,
\begin{align}
    \varphi_{X^\star}(l)=\varphi_{X}(l)\varphi_{\varepsilon}(l),
\end{align}
where $\varphi_{X},\varphi_{X^\star}$ and $\varphi_{\varepsilon}$ are the respective characteristic functions.
Then, as in \cite{DM:2022}, using the inversion formula we can get the simple estimator
\begin{align}\label{simpleEstimatorC}
    \tilde{f}_{X}(x)=\frac1{2\pi}\sum_{l=-\infty}^\infty \frac{\hat{\varphi}_{X^{\star}}(l)}{\varphi_{\varepsilon}(l)}e^{-ilx},
 \end{align}  
where $\hat{\varphi}_{X^{\star}}$ is the empirical characteristic function of the sample $X^\star_1,\ldots,X^\star_n$. We obtain
 \begin{align}   
   \tilde{f}_{X}(x)  =\frac{1}{2\pi n}\sum_{i=1}^n\Big(1+2\sum_{l=1}^m\frac{(1-\frac{l}{m+1})}{\lambda_{\varepsilon}(l)}\cos(l(x-X^\star_i))\Big).
\end{align}
According to \cite{DM:2022}, the asymptotic bias matches that of the error-free scenario given in \eqref{ISB}. We can determine the MISE's behaviour using \citet[Result 3]{DM:2022} for the asymptotic variance. For some particular distributions, the MISE can be derived explicitly.
For example, when the error follows wrapped Laplace distribution with scale parameter $\rho$, 
 its Fourier coefficients are given by $\lambda_{\varepsilon}(j)=\frac{\rho^2}{\rho^2+j^2}$. In this case,  the $\operatorname{MISE}$ is expressed as
\begin{align*}
\operatorname{MISE}\left[\tilde{f}_n(\theta;m) \right]= \frac{\theta_1(f)}{m^2}+\frac{1}{2\pi n}\Big(1+\frac{2m^5}{105\rho^4}\Big)+o\Big(\frac{1}{m^2}\Big) +o\Big(\frac{m^5}{n}\Big).
\end{align*}
Thus, the optimal bandwidth in terms of AMISE is
 \begin{align*}
     m_{opt}=\Big(\frac{42\pi}{\rho^4}\theta_1(f)\cdot n\Big)^{\frac{1}{7}}.
 \end{align*}

\section{Simulation study} \label{sec:sim-study}

The goal of this section is to explore the finite sample properties of the F\'ejer estimators for densities, both with and without measurement error, as well as for the distribution function.  In particular, we consider the mixture of distributions denoted by
$Mix(F_1,F_2,p)$ where $F_1$ and $F_2$ are either von Mises (VM) or wrapped normal (WN) circular distributions with corresponding parameters. 
The mean integrated squared error (MISE) is evaluated using $N=500$ replications, exploring a range of sample sizes $n$ and various distribution parameters.

In the case of error-free density estimation, the MISE is estimated for 
$m\in\{5,10,\sqrt{n}\}$ as well as for the asymptotically optimal choice \eqref{optimal-m}. Here,  $\theta_1(f)$ is estimated both parametrically, assuming a von Mises distribution with ML estimation of its concentration parameter, and nonparametrically using \eqref{plugin-first-derivative}, for $M(n)=2n^{1/4}$. In addition to MISE, we provide the average optimal values of $m$ for both estimation approaches, along with the theoretically optimal $m$ (assuming $\theta_1(f)$ to be known).  
Table  \ref{tab:MixWN} shows the results for mixtures of wrapped normal distributions.
Results for mixtures of VM distributions are omitted due to their similarity.

It can be seen that the average values of the nonparametric estimator of the optimal $m$ consistently overestimate the quantity of interest. In contrast, in the parametric case, the values obtained are close to the theoretical ones. The only exceptions occur for alternatives that are mixtures of distributions with antipodal mean directions. In practice, one way to improve the parametric estimator for such mixtures is to consider a two-step procedure that uses a mixture of VM distributions \citep{OC:2012}.

The same analysis is conducted in the context of the classical measurement error model.
Table \ref{tab:MixWNclassic} shows the MISE results for the classical error scenario with a wrapped Laplace distribution. The scale parameter is set to $0.2$ to cover cases with noise-to-signal ratios ranging from about 5\% to 36\%.

The Berkson model with uniform distribution is a good approximation when the error arises as a result of rounding (\cite{WW:2013}). We therefore examine the properties of different estimators in this context. We consider data that are rounded to nearest multiple of $\frac{\pi}{6}$, which implies that the errors follow a uniform $U[-\frac{\pi}{12},\frac{\pi}{12}]$ distribution. Just for comparison purposes, we also include cases with wrapped Laplace errors with scale parameters 0.1 and 0.2. The MISE of the corresponding estimators for the optimal $m$, both parametric and nonparametric as described above, are presented in Table \ref{tab:BerksonRound}. The results clearly show the importance of including the Berkson error when the data are rounded. In particular, MISE  increases, when using the error-free estimator, and when using the $WL(0.1)$ Berkson error, while it decreases for the $WL(0.2)$ and uniform error distributions. The best performance is observed under the assumption of a uniform error, which is somewhat expected given the nature of the data.

The results of the empirical study for estimating the CDF are presented in Tables \ref{tab:MixVMCDF} and \ref{tab:MixVMCDFestimatedTheta}, corresponding to cases where \(\theta_0 = -\pi\) is fixed and where \(\theta_0\) is estimated by minimizing \eqref{Cn-crit}, respectively.  

When $\theta_0$ is fixed to $-\pi$, the MISE varies significantly with the change in the location parameter, being smallest when the optimal $\theta_0$ is closest to $-\pi$. In contrast, when \(\theta_0\) is estimated, the MISE remains stable across different values of the location parameter. It is also observed that the estimated value of the optimal $\theta$ is close to the theoretical value. 
With respect to the smoothing parameter $m$, the conclusions are in line with those of the density estimation.

\section{Application to rainfall data}\label{rainfall}

We consider the data from \cite{DM:1941} (see Table \ref{tab:freq}). The data represent the number of rainfall occurrences of $1"$ or more per hour in the US from 1908 to 1937, with frequencies adjusted to overcome different month lengths, see \cite{MA:1975}. Since the data are reported monthly, the assumption of a uniform Berkson error is natural. 
We compare different density estimation methods, plotting the results using a histogram of the data along with the Fej\'er density estimator under various error scenarios. These include
 the error-free case (blue line), as well as cases where the Berkson error is modeled using different distributions: $WL(0.1)$  (orange line), $WL(0.2)$  (purple line), and $U[-\frac{\pi}{12},\frac{\pi}{12}]$  (red line). 
Two types of density estimates are considered: one where $\theta_1(f)$ is estimated parametrically, and another where it is estimated nonparametrically, using the same specifications as in the simulation study. The plots in Figure \ref{fig:rain} reveal that failing to account for Berkson error, or misspecifying it, can result in the detection of more local modes than are actually present in the data, leading to misleading conclusions.

\begin{table}[htb]
    \centering
    \begin{tabular}{c|cccccccccccc}
      Month & Jan& Feb &Mar &Apr& May& Jun &Jul&Aug&Sep&Oct&Nov&Dec  \\
      Adj. freq. &100&103&229&414&676&1248&1458&1365&924&378&199&143 
    \end{tabular}
    \caption{Rainfall data for Section \ref{rainfall}}
    \label{tab:freq}
\end{table}

\begin{figure}
  \centering
 \includegraphics[width=0.49\textwidth]{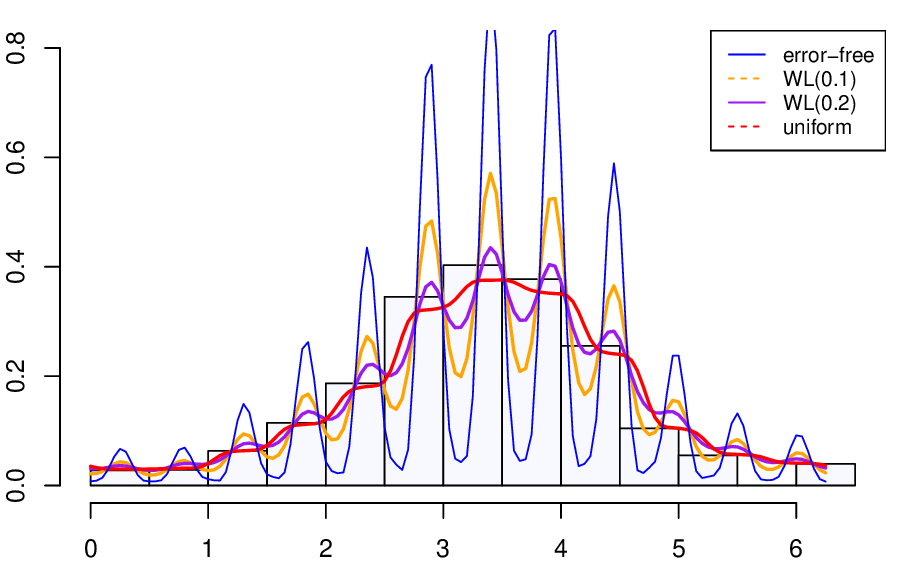}
  \includegraphics[width=0.49\textwidth]{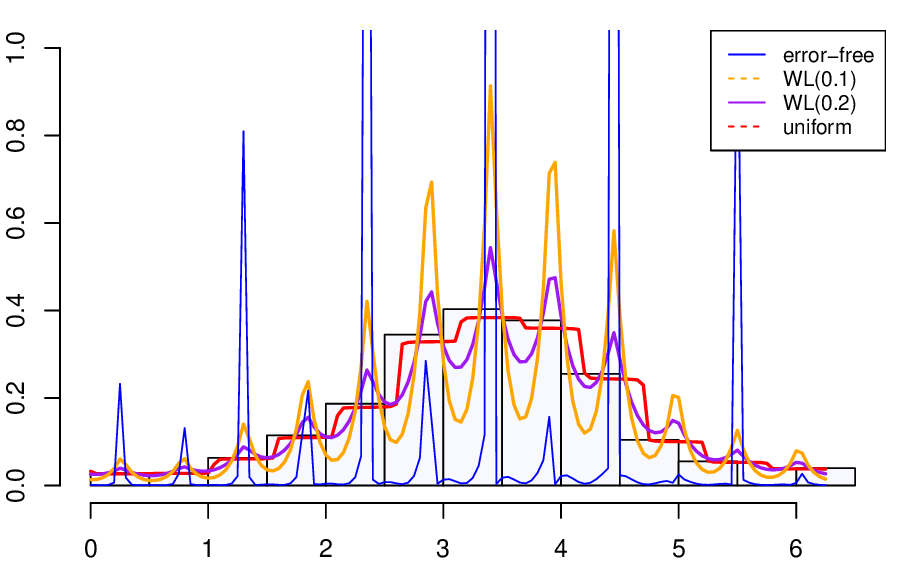}
\caption{Fej\'er density estimators with Berkson error for the rainfall data with $\theta_1(f)$ estimated parametrically (left) and nonparametrically (right) \label{fig:rain}}
\end{figure}


\section{Concluding remarks}\label{sec:conclude}

This work highlights the potential of Fej\'er polynomials as a flexible and powerful tool for density and distribution function estimation on the circle.
By applying Fej\'er’s approximation theorem, the proposed estimators inherently account for the periodic nature of circular data, addressing limitations in approaches which use Bernstein polynomials on the circle. Extending the estimators to handle measurement errors ensures applicability in realistic scenarios where data inaccuracies are common. 

Future work could investigate further refinements in bandwidth selection, particularly in relation to the data-driven selection of the lower integration limit in defining the circular CDF estimator.
Overall, the Fej\'er polynomial framework offers a promising approach for nonparametric circular data analysis.

\part*{Appendix}
\appendix

\section{Auxiliary results for subsection \ref{sec:asym-exp}} \label{app:A}

\begin{lemma} \label{lemma-harmonic}
Define
\begin{align*}
    H_m=\sum_{k=1}^m \frac{1}{k}, \quad \bar{H}_m=\sum_{k=1}^m \frac{(-1)^{k+1}}{k}, \quad
    H_m^l=\sum_{k=1}^m \frac{1}{k^l}, \quad \bar{H}_m^l=\sum_{k=1}^m \frac{(-1)^{k+1}}{k^l}, \, l\geq 2.
\end{align*} 
Then,
\begin{align*}
    H_m &= \gamma + \log m + \frac{1}{2m} + O(m^{-2}), \quad
    H_m^2 = \frac{\pi^2}{6} - \frac{1}{m} + \frac{1}{2m^2} + O(m^{-3}), \\
    H_m^3 &= \zeta(3) - \frac{1}{2m^2} + \frac{1}{2m^3} + O(m^{-4}), \quad
    H_m^4 = \frac{\pi^4}{90} - \frac{1}{3m^3} + \frac{1}{2m^4} + O(m^{-5}), \\
    \bar{H}_m &= \log 2 - \frac{1}{2m} + O(m^{-2}), \quad
    \bar{H}_m^2 = \frac{\pi^2}{12} - \frac{1}{2m^2} + O(m^{-3}), \\
    \bar{H}_m^l &= \left( 1 - \frac{1}{2^{l-1}} \right) \zeta(l) + \frac{1}{2m^l}  + O(m^{-(l+1)}), \; l\geq 3,
\end{align*} 
where $\gamma=0.5772\ldots$ is the Euler-Mascheroni constant, and $\zeta(\cdot)$ denotes the Riemann zeta function. 
\end{lemma}

\begin{proof}
The formulas for $H_m$ and $H_m^l$ follow directly from the usual asymptotic expansions based on the Euler–Maclaurin formula. The result for $\bar{H}_m$ follows from
\begin{align*}
\bar{H}_{2m} &= H_{2m} - H_{m} 
= \log(2m) + \frac{1}{4m} - \log m - \frac{1}{2m} + O(m^{-2}) \\
&= \log 2 - \frac{1}{4m} + O(m^{-2}).
\end{align*}
Similarly, the formulas for $\bar{H}_m^l$ follow from
\begin{align*}
\bar{H}_{2m}^2 &= H_{2m}^2 - \frac{1}{2} H_{m}^2 
= \frac{\pi^2}{12} - \frac{1}{8m^2} + O(m^{-3}), \\
\bar{H}_{2m}^l &= \sum_{k=1}^{2m} \frac{1}{k^l} - 2 \left( \frac{1}{2^l} + \frac{1}{4^l} + \ldots + \frac{1}{(2m)^l} \right) \\
&= H_{2m}^l + \frac{1}{2^{l-1}} H_{m}^l 
\;=\; \left( 1 - \frac{1}{2^{l-1}} \right) \zeta(l) + \frac{1}{2(2m)^l}  + O(m^{-(l+1)}).
\end{align*}
\end{proof}

\begin{proposition} \label{prop-moments}
Define
\begin{align*}
\beta(K_m) &= \int_{-\pi}^{\pi} y^2 K_m(y)dy, \quad
\gamma(K_m) = \int_{-\pi}^{\pi} |y|^3 K_m(y) dy.
\end{align*} 
Then,
\begin{align*}
\beta(K_m) &= \frac{4\log 2}{m+1} + O(m^{-3}), \quad
\gamma(K_m) = \frac{6\pi\log 2 - 21\zeta(3)/\pi}{m+1} + O(m^{-3}).
\end{align*}    
\end{proposition}

\begin{proof}
Using Lemma \ref{lemma-harmonic} and $\int_{-\pi}^{\pi} y^2 \cos(ky)dy=4\pi(-1)^k/k^2$ for integer $k$, we obtain
\begin{align*}
\beta(K_m) &= \int_{-\pi}^{\pi} y^2 \left( \frac{1}{2\pi} +  \frac{1}{\pi} \sum_{k=1}^m \left(1-\frac{k}{m+1}\right) \cos(ky) \right) dy \\
&= \frac{\pi^2}{3} +  4 \sum_{k=1}^m \left(1-\frac{k}{m+1} \right) \frac{(-1)^k}{k^2} \\
&= \frac{\pi^2}{3} - 4 \bar{H}_m^2 + \frac{4}{m+1} \bar{H}_m  \\
&= \frac{\pi^2}{3} - 4 \left( \frac{\pi^2}{12} - \frac{1}{2m^2} \right) 
+ \frac{4}{m+1} \left( \log 2 - \frac{1}{2m} \right) + O(m^{-3})   \\
&= \frac{4\log 2}{m+1} + O(m^{-3}) \\
\end{align*}
Again using Lemma \ref{lemma-harmonic} and 
\begin{align*}
    \int_{-\pi}^{\pi} |y|^3 \cos(ky)dy &= 6\left( \frac{\pi^2(-1)^k}{k^2} - \frac{2(-1)^k}{k^4} + \frac{2}{k^4} \right),
\end{align*}
we obtain
\begin{align*}
\gamma(K_m) =& \frac{\pi^3}{4} +  \frac{6}{\pi} \sum_{i=1}^m \left(1-\frac{k}{m+1} \right) \left( \frac{-\pi^2(-1)^{k+1}}{k^2} + \frac{2(-1)^{k+1}}{k^4} + \frac{2}{k^4} \right) \\
=& \frac{\pi^3}{4} + \frac{6}{\pi} \left( -\pi^2 \bar{H}_m^2 + 2\bar{H}_m^4 + 2H_m^4 \right) - \frac{6}{\pi(m+1)} \left( -\pi^2 \bar{H}_m + 2\bar{H}_m^3 + 2H_m^3 \right) \\
=& \frac{\pi^3}{4} - 6\pi \left( \frac{\pi^2}{12} - \frac{1}{2m^2}  \right) 
+ \frac{12 \cdot 15\pi^3}{8\cdot 90} \\
&- \frac{6}{\pi(m+1)} \left( -\pi^2\left(\log 2-\frac{1}{2m}\right) + \frac{7\zeta(3)}{2} \right) + O(m^{-3})\\
=& \frac{1}{m+1} \left(6\pi\log 2 - 21\zeta(3)/\pi \right) + O(m^{-3}).
\end{align*}
\end{proof}

\section{Auxiliary results for Section \ref{sec:CDF-est}}\label{app:B}

\begin{lemma} \label{lemma-m1}
Define
\[
\nu_{1,m} = 2\pi \int_{-\pi}^{\pi} y \cdot W_m(y) \cdot K_m(y) \, dy.
\]
Then,
\[
\nu_{1,m} = \frac{2 \log m}{m+1} + \frac{2}{m+1} \left( \log 2 + \gamma \right) + O(m^{-2}),
\]
where $\gamma=0.5772\ldots$ is the Euler-Mascheroni constant. 
\end{lemma}

\begin{proof} {\small
We need to evaluate
\begin{align*}
\frac{\nu_{1,m}}{2\pi} &= \int_{-\pi}^{\pi} y \left( \frac{y + \pi}{2\pi} + \frac{1}{\pi} \sum_{k=1}^m \frac{1 - \frac{k}{m+1}}{k} \sin(ky) \right) \left( \frac{1}{2\pi} + \frac{1}{\pi} \sum_{l=1}^m \left(1 - \frac{l}{m+1}\right) \cos(ly) \right) dy.
\end{align*}
Expanding this, we have four terms to consider. The first two terms are
\[
\int_{-\pi}^{\pi} y \cdot \frac{y + \pi}{2\pi} \cdot \frac{1}{2\pi} \, dy = \frac{\pi}{6}
\] 
and
\begin{align*}
\int_{-\pi}^{\pi} & y \cdot \frac{y + \pi}{2\pi} \cdot \frac{1}{\pi} \sum_{k=1}^m \left(1 - \frac{k}{m+1}\right) \cos(ky) \, dy 
\ = \ \frac{1}{2\pi^2} \sum_{k=1}^m \left(1 - \frac{k}{m+1}\right) \int_{-\pi}^{\pi} y^2 \cos(ky) \, dy \\
& = \frac{1}{2\pi^2} \sum_{k=1}^m \left(1 - \frac{k}{m+1}\right) 4(-1)^k \pi/k^2 
\ = \ \frac{2}{\pi} \sum_{k=1}^m \left(\frac{1}{k^2} - \frac{1}{(m+1)k}\right) (-1)^k.
\end{align*}
Similarly, the third term is
\begin{align*}
\int_{-\pi}^{\pi} & y \cdot \frac{1}{\pi} \sum_{k=1}^m \frac{1 - \frac{k}{m+1}}{k} \sin(ky) \cdot \frac{1}{2\pi} \, dy 
= \frac{1}{2\pi^2} \sum_{k=1}^m \frac{1 - \frac{k}{m+1}}{k} \int_{-\pi}^{\pi} y \sin(ky) \, dy \\
=& \frac{1}{2\pi^2} \sum_{k=1}^m \frac{1 - \frac{k}{m+1}}{k} (-2)(-1)^k\pi/k 
= \frac{1}{\pi} \sum_{k=1}^m \left(\frac{1}{k^2} - \frac{1}{(m+1)k}\right) (-1)^{k+1}.
\end{align*}
Finally, we obtain for the last term:
\begin{align*}
\int_{-\pi}^{\pi} & y \cdot \frac{1}{\pi} \sum_{k=1}^m \frac{1 - \frac{k}{m+1}}{k} \sin(ky) \cdot \frac{1}{\pi} \sum_{l=1}^m \left(1 - \frac{l}{m+1}\right) \cos(ly) \, dy \\
=& \frac{1}{\pi^2} \sum_{k=1}^m \sum_{l=1}^m \left(1-\frac{k}{m+1}\right) \left(1-\frac{l}{m+1}\right) \frac{1}{k} \int_{-\pi}^{\pi} y \sin(ky) \cos(ly) \, dy \\
=& \frac{1}{\pi^2} \sum_{k\neq l}^m \left(1-\frac{k}{m+1}\right) \left(1-\frac{l}{m+1}\right) \frac{1}{k} (-\pi) (-1)^{k+l} \left( \frac{1}{k-l} + \frac{1}{k+l} \right) \\
& + \frac{1}{\pi^2} \sum_{k=1}^m \left(1-\frac{k}{m+1}\right)^2 \frac{1}{k}\ \frac{(-\pi)}{2k}
\\
=& \frac{-2}{\pi} \sum_{k\neq l}^m \left( 1 - \frac{k+l}{m+1} + \frac{kl}{(m+1)^2} \right) \frac{(-1)^{k+l}}{k^2 - l^2} 
 - \frac{1}{2\pi} \sum_{k=1}^m \left(1 - \frac{k}{m+1}\right)^2 \frac{1}{k^2} \\
=& -\frac{1}{2\pi} \sum_{k=1}^m \left(\frac{1}{k^2} - \frac{2}{(m+1)k} + \frac{1}{(m+1)^2}\right).
\end{align*}
Hence,
\begin{align*}
\nu_{1,m} =& \frac{\pi^2}{3} + 4 \left( -\bar{H}_m^2 + \frac{1}{m+1}\bar{H}_m \right)
	+ 2 \left( \bar{H}_m^2 - \frac{1}{m+1}\bar{H}_m \right)
	- \left( H_m^2 - \frac{2}{m+1} H_m + \frac{m}{(m+1)^2} \right) \\
	=& \frac{\pi^2}{3} + 4 \left( \frac{-\pi^2}{12} + \frac{1}{2m^2} + \frac{1}{m+1} \left(\log 2-\frac{1}{2m}\right) \right)
	+ 2 \left( \frac{\pi^2}{12} - \frac{1}{2m^2} - \frac{1}{m+1} \left(\log 2-\frac{1}{2m}\right) \right) \\
	 & - \left( \frac{\pi^2}{6} - \frac{1}{m} + \frac{1}{2m^2} - \frac{2}{m+1} \left( \gamma + \log m + \frac{1}{2m} \right)
	 	 + \frac{m}{(m+1)^2} \right) + O(m^{-3}) \\
	=& \frac{2 \log m}{m+1} + \frac{2}{m+1} \left( \log 2 + \gamma \right) + O(m^{-2}).
\end{align*}
} \end{proof}

\begin{lemma} \label{lemma-m3}
Define
\[
\nu_{3,m} = \int_{-\pi}^{\pi} y^3 \cdot W_m(y) \cdot K_m(y) \, dy.
\]
Then,
\[
\nu_{3,m} = O(m^{-1}).  
\]
\end{lemma} 

\begin{proof} {\small
We need to evaluate
\begin{align*}
\nu_{3,m} &= \int_{-\pi}^{\pi} y^3 \left( \frac{y + \pi}{2\pi} + \frac{1}{\pi} \sum_{k=1}^m \frac{1 - \frac{k}{m+1}}{k} \sin(ky) \right) \left( \frac{1}{2\pi} + \frac{1}{\pi} \sum_{l=1}^m \left(1 - \frac{l}{m+1}\right) \cos(ly) \right) dy.
\end{align*}
Expanding this, we have four terms to consider. The first two terms are
\[
\int_{-\pi}^{\pi} y^3 \cdot \frac{y + \pi}{2\pi} \cdot \frac{1}{2\pi} \, dy = \frac{\pi^3}{10}
\] 
and
\begin{align*}
\int_{-\pi}^{\pi} & y^3 \cdot \frac{y + \pi}{2\pi} \cdot \frac{1}{\pi} \sum_{k=1}^m \left(1 - \frac{k}{m+1}\right) \cos(ky) \, dy 
\ = \ \frac{1}{2\pi^2} \sum_{k=1}^m \left(1 - \frac{k}{m+1}\right) \int_{-\pi}^{\pi} y^4 \cos(ky) \, dy \\
& = \frac{1}{2\pi^2} \sum_{k=1}^m \left(1 - \frac{k}{m+1}\right) \frac{(-1)^k 8\pi (k^2 \pi^2-6)}{k^4}
\ = \ \frac{4}{\pi} \sum_{k=1}^m \left(1-\frac{k}{m+1}\right) \frac{(-1)^k (k^2 \pi^2-6)}{k^4} \\
&= \frac{4}{\pi} \left( -\pi^2\bar{H}_m^2 + 6\bar{H}_m^4 \right) + O(m^{-1})
\ = \ \frac{4}{\pi} \left( -\pi^2 \cdot \frac{\pi^2}{12} + 6 \cdot \frac{7}{8} \ \zeta(4) \right) + O(m^{-1}) \\
&= \frac{-\pi^3}{10} + O(m^{-1}).
\end{align*}
Similarly, the third term is
\begin{align*}
\int_{-\pi}^{\pi} & y^3 \cdot \frac{1}{\pi} \sum_{k=1}^m \frac{1 - \frac{k}{m+1}}{k} \sin(ky) \cdot \frac{1}{2\pi} \, dy 
= \frac{1}{2\pi^2} \sum_{k=1}^m \frac{1 - \frac{k}{m+1}}{k} \int_{-\pi}^{\pi} y^3 \sin(ky) \, dy \\
=& \frac{1}{2\pi^2} \sum_{k=1}^m \frac{1 - \frac{k}{m+1}}{k} \frac{(-2) (-1)^k \pi (k^2 \pi^2-6)}{k^3} 
= \frac{-1}{\pi} \sum_{k=1}^m \left(1-\frac{k}{m+1}\right) \frac{(-1)^{k} (k^2 \pi^2-6)}{k^4} \\
&= \frac{\pi^3}{40} + O(m^{-1}).
\end{align*}
Finally, we obtain for the last term:
\begin{align*}
\int_{-\pi}^{\pi} & y^3 \cdot \frac{1}{\pi} \sum_{k=1}^m \frac{1 - \frac{k}{m+1}}{k} \sin(ky) \cdot \frac{1}{\pi} \sum_{l=1}^m \left(1 - \frac{l}{m+1}\right) \cos(ly) \, dy \\
=& \frac{1}{\pi^2} \sum_{k=1}^m \sum_{l=1}^m \left(1-\frac{k}{m+1}\right) \left(1-\frac{l}{m+1}\right) \frac{1}{k} \int_{-\pi}^{\pi} y^3 \sin(ky) \cos(ly) \, dy \\
=& \frac{1}{\pi^2} \sum_{k=1}^m \left(1-\frac{k}{m+1}\right)^2 \frac{1}{k}\ \frac{6k\pi-4k^3\pi^3}{8 k^4} \\
 &+ \frac{1}{\pi^2} \sum_{k\neq l}^m \left(1-\frac{k}{m+1}\right) \left(1-\frac{l}{m+1}\right) \frac{1}{k} \pi (-1)^{k+l} \left( \frac{6-(k-l)^2 \pi^2}{(k-l)^3} + \frac{6-(k+l)^2 \pi^2}{(k+l)^3} \right) \\ 
=&  S_{m,1}+S_{m,2},
\end{align*}
say. For the first sum, we have 
\begin{align*}
S_{m,1} =& 
\frac{1}{4\pi} \sum_{k=1}^m \left(1-\frac{k}{m+1}\right)^2 \ \frac{3-2k^2\pi^2}{k^4} 
\ = \ \frac{1}{4\pi} \left(3 H_m^4 - 2 \pi^2 H_m^2 + \frac{4\pi^2}{m+1} H_m^1 \right) + O(m^{-1}) \\
=& -\frac{3\pi^3}{40} + \frac{\pi\log(m)}{m+1}  + O(m^{-1})
\end{align*}
For the second sum, we get by symmetry arguments
\begin{align*}
S_{m,2} &=  \frac{1}{\pi} \sum_{k\neq l}^m \left(1-\frac{k+l}{m+1}+\frac{kl}{(m+1)^2} \right) (-1)^{k+l} \left( \frac{-2\pi^2}{k^2 - l^2} + \frac{12k^2+36l^2}{(k-l)^3(k+l)^3} \right) \\
=& \frac{1}{\pi} \sum_{k\neq l}^m \left(1-\frac{k+l}{m+1}+\frac{kl}{(m+1)^2} \right) (-1)^{k+l} \frac{24l^2}{(k-l)^3(k+l)^3}.
\end{align*}
We claim that 
\begin{align*}
S_{m,2} =& \frac{2\pi^3}{40} - \frac{\pi\log(m)}{m+1}  + O(m^{-1}).
\end{align*}
This could only be verified numerically, as shown in the following table:

\begin{center}
\begin{tabular}{rrrr}  \hline
 $m$ & $m\left(S_{m,1}+\frac{3\pi^3}{40}-\frac{\pi\log(m)}{m+1}\right)$ 
 & $m\left(S_{m,2}-\frac{2\pi^3}{40}+\frac{\pi\log(m)}{m+1}\right)$ 
 & $m\left(S_{m,1}+S_{m,2}+\frac{\pi^3}{40}\right)$ \\  \hline
   50 & 1.29874 & -0.18366 & 1.11508 \\ 
  100 & 1.26966 & -0.13840 & 1.13125 \\ 
  200 & 1.25469 & -0.11518 & 1.13951 \\ 
  400 & 1.24710 & -0.10342 & 1.14368 \\ 
  800 & 1.24328 & -0.09750 & 1.14578 \\ 
 1600 & 1.24136 & -0.09453 & 1.14683 \\ 
 3200 & 1.24040 & -0.09305 & 1.14735 \\ 
 6400 & 1.23992 & -0.09230 & 1.14762 \\ 
12800 & 1.23968 & -0.09193 & 1.14775 \\  \hline
\end{tabular}
\end{center}

So the fourth term is $S_{m,1}+S_{m,2}=\pi^3/40+O(m^{-1})$, and summing over all the terms gives the desired result.
}
\end{proof}

\begin{lemma} \label{lemma-gamma1-2}
\begin{enumerate}
\item[a)] 
It holds that 
\[
S_{m,1}(x) = \sum_{k=1}^{m} \frac{\sin(kx)}{k} \ = \ \frac{\pi-x}{2} + O\left(m^{-1}\right),
\]
and
\[
S_{m,2}(x) = \sum_{k=1}^{m} (-1)^k \, \frac{\sin(kx)}{k} \ = \ -\frac{x}{2} + O\left(m^{-1}\right).
\]
\item[b)]
Define
\[
\gamma_{m}(\theta) = \sum_{k=1}^{m} (-1)^k \left(1-\frac{k}{m+1}\right)^2 \frac{1}{k} \sin(k\theta)
\]
and
\begin{align*}
\gamma_{1,m}(\theta) &= \theta + 2 \gamma_{m}(\theta), \\
\gamma_{2,m}(\theta) &= \frac{1}{2} + \frac{1}{\pi} \left( -\sum_{k=1}^{m} \left(\frac{1}{k}-\frac{1}{m+1}\right) \sin(k\theta) + \gamma_{m}(\theta) \right).
\end{align*}
Then,
\[
\gamma_{1,m}(\theta) = O(m^{-1}) \quad \text{and} \quad  \gamma_{2,m}(\theta) = O(m^{-1}).
\]
\end{enumerate}
\end{lemma}

\begin{proof}
\begin{enumerate}
\item[a)] 
It holds that
\[
S_1(x)=\sum_{k=1}^{\infty} \frac{\sin(kx)}{k} = \frac{\pi-x}{2} \quad \text{for } 0<x< 2\pi
\]
and 
\[
S_2(x)=\sum_{k=1}^{\infty} (-1)^k \, \frac{\sin(kx)}{k} = -\frac{x}{2} \quad \text{for } -\pi<x<\pi
\]
\citep[1.441]{GR:1986}. 
Abel's convergence test states that if a sequence of positive real numbers $(a_k)$ is decreasing monotonically to zero, then the power series
$\sum _{k=0}^{\infty }a_{k}z^{k}$ converges everywhere on the closed unit circle, except when $z=1$. The usual proof by Abel summation shows that
\[
\left| \sum _{k=m}^{n} a_{k}z^{k} \right| \leq a_m \, \frac{4}{|1-z|}.
\]
Choosing $a_k=1/k$ and $z=e^{ikx}$ yields
\begin{align*}
  |S_{1}(x)-S_{m,1}(x)| &= \left| \sum _{k=m}^{\infty} \frac{1}{k} \Imag(e^{ikx}) \right| \leq \frac{1}{m} \frac{4}{|1-e^{ikx}|} = O(m^{-1}).
\end{align*}
Noting that $\Imag(e^{ik(x+\pi)})=(-1)^k \sin(kx)$, the same argument yields the second assertion. 
\item[b)] 
First, note that 
\[
\left| S_{m,3}(x) \right| = \left|\sum_{k=1}^{m} \sin(kx) \right| \leq c_3(x), \quad 
\left| S_{m,4}(x) \right| = \left|\sum_{k=1}^{m} (-1)^k \sin(kx) \right| \leq c_4(x)
\]
and 
\[
\left| S_{m,5}(x) \right| = \left|\sum_{k=1}^{m} (-1)^k k\sin(kx) \right| \leq m \, c_5(x)
\]
for positive constants $c_i(x),i=1,2,3,$ not depending on $m$ \citep[1.342,1.352,1.353]{GR:1986}. Hence,
\[
\gamma_{m}(\theta) = S_{m,2}(\theta) - \frac{2}{m+1} S_{m,4}(\theta) + \frac{1}{(m+1)^2} S_{m,5}(\theta) = -\frac{\theta}{2} + O(m^{-1}).
\]
It follows directly that $\gamma_{1,m}(\theta) = O(m^{-1})$. Second,
\[
\pi \gamma_{2,m}(\theta) = \frac{\pi-\theta}{2} - S_{m,1}(\theta) + \frac{2}{m+1} S_{m,3}(\theta) + \frac{1}{2} \gamma_{1,m}(\theta) = O(m^{-1}),
\]
which concludes the proof
\end{enumerate}
\end{proof}

\begin{lemma} \label{prop-high-mom}
It holds that 
\begin{align*}
m_4(K_m) &= \int_{-\pi}^{\pi} y^4 K_m(y)dy \ = \frac{8\pi^2\log(2)-36\zeta(3)}{m} + O(m^{-2}).
\end{align*} 
\end{lemma}

\begin{proof}
Using $\int_{-\pi}^{\pi} y^4 \cos(ky)dy=8\pi(-1)^k(\pi^2k^2-6)/k^4$ for integer $k$, we obtain
\begin{align*}
m_4(K_m) &= \int_{-\pi}^{\pi} y^4 \left( \frac{1}{2\pi} +  \frac{1}{\pi} \sum_{k=1}^m \left(1-\frac{k}{m+1}\right) \cos(ky) \right) dy \\
&= \frac{\pi^4}{5} +  8 \sum_{i=1}^m \left(1-\frac{k}{m+1} \right) \frac{(-1)^k(\pi^2k^2-6)}{k^4} \\
&= \frac{\pi^4}{5} - 8\pi^2 \left(\bar{H}_m^2 - \frac{1}{m+1} \bar{H}_m^1 \right) + 48 \left( \bar{H}_m^4 - \frac{1}{m+1} \bar{H}_m^3 \right) + O(m^{-2}).
\end{align*}
Inserting the expansions of $\bar{H}_m^l, l=1,\ldots,4,$ given in Lemma \ref{lemma-harmonic} proves the claim.
\end{proof}

\section*{Acknowledgement}
The research was funded by the bilateral cooperation project ”Modeling complex data - Selection and Specification” between the Federal Republic of Germany and the Republic of Serbia (337-00-19/2023-01/6) by  the Ministry of Science, Technological Development and Innovations of the Republic of Serbia and the Federal Ministry of Education and Research of Germany. The work of B. Milošević and M. Obradović is supported by the Ministry of Science, Technological Development and Innovations of the Republic of Serbia (the contract 451-03-66/2024-03/200104). B. Milošević  is also supported by the COST action CA21163 - Text, functional and other high-dimensional data in econometrics: New models, methods, applications (HiTEc).
\bibliographystyle{apalike}
\bibliography{Fejer-biblio}

\clearpage
\renewcommand{\thefootnote}{\fnsymbol{footnote}}

\begin{table}
    \centering
        \resizebox{\textwidth}{!}{
    \begin{tabular}{l|c|c|c|c|c|c|c|c|c|c}
    &&\multicolumn{5}{|c|}{$\MISE(f;\hat{f}_{m,n})$}
    &\multicolumn{2}{|c|}{Average}&\\\hline
     Distribution&$n$    &$m=5$ &$m=10$&$m=[\sqrt{n}]$ &$m=m_{OP}$&$m=m_{ON}$&$m_{OP}$ &$m_{ON}$ &$m_{TH}$ \\\hline
     $WN(0,0.75)$    &50   &$3.36\cdot 10^{-4}$ &$4.78\cdot 10^{-4}$&$3.35\cdot 10^{-4}$ &$3.54\cdot 10^{-4}$&$5.44\cdot 10^{-4}$&7.41&10.1&6.73 \\
      $WN(0,0.9)$  &50  &$2.26\cdot 10^{-3}$ &$8.77\cdot 10^{-4}$&$1.18\cdot 10^{-3}$&$9.26\cdot 10^{-4}$&$1.10\cdot 10^{-3}$&11.7&14.8&11.1\\
     $Mix(WN(0,0.9),WN(\pi/2,0.75),0.5)$    &50  &$3.71\cdot 10^{-4}$  &$5.78\cdot 10^{-4}$&$3.99\cdot 10^{-4}$ &$3.86\cdot 10^{-4}$   &$6.19\cdot 10^{-4}$&$5.50$ &9.78& 6.69
     \\
     $Mix(WN(0,0.9),WN(\pi/2,0.9),0.5)$    &50&$6.13\cdot 10^{-4}$  &$6.67\cdot 10^{-4}$&$5.34\cdot 10^{-4}$ &$5.61\cdot 10^{-4}$   &$7.59\cdot 10^{-4}$ &$5.99$ &10.3& 7.98
     \\
     $Mix(WN(0,0.9),WN(\pi/2,0.75),0.2)$  &50 &$2.10\cdot 10^{-4}$  &$4.64\cdot 10^{-4}$&$2.72\cdot 10^{-4}$ &$2.44\cdot 10^{-4}$   &$4.00\cdot 10^{-4}$&5.94  &7.92& 5.57
     \\
     $Mix(WN(0,0.9),WN(\pi/2,0.75),0.8)$   &50  &$1.16\cdot 10^{-3}$  &$7.75\cdot 10^{-4}$&$7.88\cdot 10^{-4}$ &$8.38\cdot 10^{-4}$   &$1.01\cdot 10^{-3}$&7.21  &12.9& 9.34
     \\
     $Mix(WN(0,0.75),WN(\pi/2,0.75),0.5)$    &50  &$1.80\cdot 10^{-4}$  &$4.91\cdot 10^{-4}$&$2.74\cdot 10^{-4}$ &$1.88\cdot 10^{-5}$   &$3.84\cdot 10^{-4}$ &5.00&7.79& 4.52\\
     $Mix(WN(0,0.75),WN(\pi/2,0.75),0.2)$    &50  &$2.20\cdot 10^{-4}$  &$4.67\cdot 10^{-4}$&$2.79\cdot 10^{-4}$ &$2.51\cdot 10^{-4}$ &  $4.01\cdot 10^{-4}$ &5.87&7.75& 5.53
     \\
     $Mix(WN(0,0.75),WN(\pi,0.75),0.5)$    &50   &$2.33\cdot 10^{-4}$  &$5.45\cdot 10^{-4}$&$3.25\cdot 10^{-4}$ &$6.61\cdot 10^{-5}$ & $4.90\cdot 10^{-4}$ &$1.67$&8.74&4.94
     \\\hline 
     
   $WN(0,0.75)$    &200   &$1.33\cdot 10^{-4}$ &$6.18\cdot 10^{-5}$&$7.58\cdot 10^{-5}$ &$6.61\cdot 10^{-5}$&$9.15\cdot 10^{-5}$&11.6&15.5&10.7 \\
      $WN(0,0.9)$  &200  &$1.80\cdot 10^{-3}$ &$2.85\cdot 10^{-4}$&$1.80\cdot 10^{-4}$&$1.75\cdot 10^{-4}$&$2.11\cdot 10^{-3}$&18.2&23.5&17.6\\
     $Mix(WN(0,0.9),WN(\pi/2,0.75),0.5)$    &200  &$1.28\cdot 10^{-4}$  &$6.33\cdot 10^{-5}$&$7.92\cdot 10^{-5}$ &$6.61\cdot 10^{-5}$   &$9.35\cdot 10^{-6}$&$8.56$ &15.2& 10.6
     \\
     $Mix(WN(0,0.9),WN(\pi/2,0.9),0.5)$    &200&$3.93\cdot 10^{-4}$  &$9.73\cdot 10^{-5}$&$9.91\cdot 10^{-5}$ &$1.03\cdot 10^{-4}$   &$1.18\cdot 10^{-4}$ &$9.34$ &16.2&12.7
     \\
     $Mix(WN(0,0.9),WN(\pi/2,0.75),0.2)$  &200 &$5.63\cdot 10^{-5}$  &$4.36\cdot 10^{-5}$&$6.41\cdot 10^{-5}$ &$4.20\cdot 10^{-5}$   &$5.64\cdot 10^{-5}$&9.18  &11.5& 8.85
     \\
     $Mix(WN(0,0.9),WN(\pi/2,0.75),0.8)$   &200  &$7.03\cdot 10^{-4}$  &$1.56\cdot 10^{-4}$&$1.25\cdot 10^{-4}$ &$1.44\cdot 10^{-4}$   &$1.62\cdot 10^{-4}$&11.2  &20.5& 14.8
     \\
     $Mix(WN(0,0.75),WN(\pi/2,0.75),0.5)$    &200  &$2.79\cdot 10^{-5}$  &$3.65\cdot 10^{-5}$&$6.12\cdot 10^{-5}$ &$2.90\cdot 10^{-5}$   &$5.02\cdot 10^{-5}$ &7.79&11.3& 7.18\\
     $Mix(WN(0,0.75),WN(\pi/2,0.75),0.2)$    &200  &$5.74\cdot 10^{-5}$  &$4.44\cdot 10^{-5}$&$6.46\cdot 10^{-5}$ &$4.30\cdot 10^{-5}$ &  $5.80\cdot 10^{-4}$ &9.08&11.5&8.77
     \\
     $Mix(WN(0,0.75),WN(\pi,0.75),0.5)$    &200   &$4.24\cdot 10^{-5}$  &$4.47\cdot 10^{-5}$&$6.93\cdot 10^{-5}$ &$5.67\cdot 10^{-4}$ & $6.67\cdot 10^{-5}$ &$1.65$&13.0&7.84
     \\\hline 
    \end{tabular}
    }
    \caption{MISE for density estimation of mixtures of wrapped normal distributions}
    \label{tab:MixWN}
\end{table}

\begin{table}
    \centering
        \resizebox{\textwidth}{!}{
    \begin{tabular}{l|c|c|c|c|c|c|c|c|c|c}
    &&\multicolumn{5}{|c|}{$\operatorname{MISE}\left[\tilde{f}_n(\theta;m) \right]$}&\multicolumn{2}{|c|}{Average}&\\\hline
     Distribution&$n$    &$m=5$ &$m=10$&$m=[\sqrt{n}]$ &$m=m_{OP}$&$m=m_{ON}$&$m_{OP}$ &$m_{ON}$ &$m_{TH}$ \\\hline
     $WN(0,0.75)$    &50   &$5.26\cdot 10^{-4}$ &$4.18\cdot 10^{-3}$&$9.82\cdot 10^{-4}$ &$1.33\cdot 10^{-3}$&$2.63\cdot 10^{-3}$&7.62&8.76& 7.50 \\
      $WN(0,0.9)$  &50  &$2.90\cdot 10^{-3}$ &$5.41\cdot 10^{-3}$&$2.41\cdot 10^{-3}$&$3.79\cdot 10^{-3}$&$5.86\cdot 10^{-3}$&8.97&10.0&9.29\\
     $Mix(WN(0,0.9),WN(\pi/2,0.75),0.5)$    &50  &$6.32\cdot 10^{-4}$  &$4.32\cdot 10^{-3}$&$1.15\cdot 10^{-3}$ &$1.06\cdot 10^{-3}$   &$2.45\cdot 10^{-3}$&$6.77$ &8.45& 7.48
     \\
     $Mix(WN(0,0.9),WN(\pi/2,0.9),0.5)$    &50&$9.42\cdot 10^{-4}$  &$4.49\cdot 10^{-3}$&$1.37\cdot 10^{-3}$ &$1.36\cdot 10^{-3}$   &$2.65\cdot 10^{-4}$ &$6.98$ &8.48& 8.07
     \\
     $Mix(WN(0,0.9),WN(\pi/2,0.75),0.2)$  &50 &$4.16\cdot 10^{-4}$  &$4.12\cdot 10^{-3}$&$9.33\cdot 10^{-4}$ &$9.10\cdot 10^{-4}$   &$1.86\cdot 10^{-3}$&6.95  &7.92& 6.92
     \\
     $Mix(WN(0,0.9),WN(\pi/2,0.75),0.8)$   &50  &$1.56\cdot 10^{-3}$  &$4.74\cdot 10^{-3}$&$1.69\cdot 10^{-3}$ &$1.95\cdot 10^{-4}$   &$4.16\cdot 10^{-3}$&7.51  &9.44& 8.63     \\
     $Mix(WN(0,0.75),WN(\pi/2,0.75),0.5)$    &50  &$3.32\cdot 10^{-4}$  &$4.25\cdot 10^{-3}$&$8.85\cdot 10^{-4}$ &$6.93\cdot 10^{-4}$   &$1.88\cdot 10^{-3}$ &6.45&7.97& 6.32\\
     $Mix(WN(0,0.75),WN(\pi/2,0.75),0.2)$    &50  &$4.25\cdot 10^{-4}$  &$4.08\cdot 10^{-3}$&$9.40\cdot 10^{-4}$ &$8.99\cdot 10^{-4}$ &  $1.87\cdot 10^{-3}$ &6.91&7.91& 6.89
     \\
     $Mix(WN(0,0.75),WN(\pi,0.75),0.5)$    &50   &$4.02\cdot 10^{-4}$  &$4.11\cdot 10^{-3}$&$9.30\cdot 10^{-4}$ &$3.51\cdot 10^{-4}$ & $1.98\cdot 10^{-3}$ &$4.07$&8.24& 6.57
     \\\hline 
     
   $WN(0,0.75)$    &200   &$1.64\cdot 10^{-4}$ &$3.22\cdot 10^{-4}$&$1.89\cdot 10^{-3}$ &$2.31\cdot 10^{-4}$&$4.49\cdot 10^{-4}$&9.08&10.5& 9.15\\
      $WN(0,0.9)$  &200  &$1.94\cdot 10^{-3}$ &$7.55\cdot 10^{-4}$&$2.38\cdot 10^{-3}$&$8.86\cdot 10^{-4}$&$1.30\cdot 10^{-3}$&10.9&12.1& 11.3\\
     $Mix(WN(0,0.9),WN(\pi/2,0.75),0.5)$    &200  &$1.63\cdot 10^{-4}$  &$3.17\cdot 10^{-4}$&$1.84\cdot 10^{-3}$ &$1.71\cdot 10^{-4}$   &$3.71\cdot 10^{-4}$&8.05 &10.1& 9.11 
     \\
     $Mix(WN(0,0.9),WN(\pi/2,0.9),0.5)$    &200&$3.46\cdot 10^{-4}$  &$3.72\cdot 10^{-4}$&$1.87\cdot 10^{-3}$ &$2.56\cdot 10^{-4}$   &$4.38\cdot 10^{-4}$ &8.39 &10.2& 9.85
     \\
     $Mix(WN(0,0.9),WN(\pi/2,0.75),0.2)$  &200 &$7.68\cdot 10^{-5}$  &$2.81\cdot 10^{-4}$&$1.90\cdot 10^{-3}$ &$1.41\cdot 10^{-4}$   &$2.61\cdot 10^{-4}$&8.30  &9.25& 8.43
     \\
     $Mix(WN(0,0.9),WN(\pi/2,0.75),0.8)$   &200  &$7.95\cdot 10^{-4}$  &$5.17\cdot 10^{-4}$&$2.07\cdot 10^{-3}$ &$4.34\cdot 10^{-4}$   &$8.26\cdot 10^{-4}$&9.03 &11.4& 10.5
     \\
     $Mix(WN(0,0.75),WN(\pi/2,0.75),0.5)$    &200  &$4.25\cdot 10^{-5}$  &$2.66\cdot 10^{-4}$&$1.83\cdot 10^{-3}$ &$1.01\cdot 10^{-4}$   &$2.45\cdot 10^{-4}$ &7.95&9.29&7.71 \\
     $Mix(WN(0,0.75),WN(\pi/2,0.75),0.2)$    &200  &$7.78\cdot 10^{-5}$  &$2.94\cdot 10^{-4}$&$1.95\cdot 10^{-3}$ &$1.41\cdot 10^{-4}$ &  $2.90\cdot 10^{-4}$ &8.23&9.31& 8.40
     \\
     $Mix(WN(0,0.75),WN(\pi,0.75),0.5)$    &200   &$5.49\cdot 10^{-5}$  &$2.59\cdot 10^{-4}$&$1.76\cdot 10^{-3}$ &$7.22\cdot 10^{-5}$ & $2.65\cdot 10^{-4}$ &4.00&9.72& 8.01
     \\\hline 
    \end{tabular}
    }
    \caption{MISE for density estimation of mixtures of wrapped normal distributions in the presence of classical error with wrapped Laplace WL(0.2) distribution}
    \label{tab:MixWNclassic}
\end{table}
\newpage

\begin{table}
    \centering
    \resizebox{\textwidth}{!}{
    \begin{tabular}{c|c|cc|cc|cc|cc|cc}
    &&\multicolumn{8}{c|}{Error distribution}&\\\hline
    &    &\multicolumn{2}{c|}{ none}&\multicolumn{2}{c|}{$WL(0.1)$}&\multicolumn{2}{c|}{$WL(0.2)$} &\multicolumn{2}{c|}{$U[-\frac{\pi}{12},-\frac{\pi}{12}]$} &\multicolumn{2}{|c}{Average} \\
& $n$ & param & nonpar & param & nonpar & param & nonpar & param & nonpar &$m_{OP}$&$m_{ON}$ 
\\\hline    
    \multirow{4}{*}{\rotatebox{90}{$VM(\pi,5)$}}&50&$2.68\cdot10^{-3}$&$1.33\cdot10^{-2}$&$1.68\cdot10^{-3}$&$2.41\cdot10^{-3}$&$2.25\cdot10^{-3}$&$2.01\cdot10^{-3}$&$1.54\cdot10^{-3}$&$1.45\cdot10^{-3}$&10.9&13.8\\
    &100&$5.81\cdot10^{-3}$&$3.58\cdot10^{-2}$&$1.19\cdot10^{-3}$&$2.71\cdot10^{-3}$&$1.33\cdot10^{-3}$&$1.25\cdot10^{-3}$&$6.93\cdot10^{-4}$&$6.44\cdot10^{-4}$&13.6&17.6\\
    &200&$2.28\cdot10^{-2}$&$8.99\cdot10^{-2}$&$1.92\cdot10^{-3}$&$4.47\cdot10^{-3}$&$1.12\cdot10^{-2}$&$1.13\cdot10^{-3}$&$4.70\cdot10^{-4}$&$4.54\cdot10^{-4}$&16.9&22.1\\
&500&$9.13\cdot10^{-2}$&$9.62\cdot10^{-1}$&$4.36\cdot10^{-3}$&$1.67\cdot10^{-2}$&$9.89\cdot10^{-4}$&$1.25\cdot10^{-3}$&$3.15\cdot10^{-4}$&$3.98\cdot10^{-4}$&22.9&45.3\\\hline
    \multirow{4}{*}{\rotatebox{90}{$VM(0,1)$}}&50&$1.69\cdot10^{-4}$&$3.80\cdot10^{-4}$&$1.62\cdot10^{-4}$&$2.72\cdot10^{-4}$&$1.56\cdot10^{-4}$&$1.81\cdot10^{-4}$&$1.61\cdot10^{-4}$&$2.50\cdot10^{-4}$&4.28&7.48\\
    &100&$6.86\cdot10^{-5}$&$4.08\cdot10^{-4}$&$6.57\cdot10^{-5}$&$1.30\cdot10^{-4}$&$6.51\cdot10^{-5}$&$6.84\cdot10^{-5}$&$6.50\cdot10^{-5}$&$9.63\cdot10^{-4}$&5.31&9.82\\
    
    &200&$2.95\cdot10^{-5}$&$4.58\cdot10^{-4}$&$2.85\cdot10^{-5}$&$6.14\cdot10^{-5}$&$3.03\cdot10^{-5}$&$2.72\cdot10^{-5}$&$2.81\cdot10^{-5}$&$3.15\cdot10^{-5}$&6.58&11.4\\
&500&$9.02\cdot10^{-6}$&$1.77\cdot10^{-3}$&$8.33\cdot10^{-6}$&$8.55\cdot10^{-5}$&$9.55\cdot10^{-6}$&$1.11\cdot10^{-5}$&$8.09\cdot10^{-6}$&$8.29\cdot10^{-6}$&8.99&15.0\\\hline
    \multirow{4}{*}{\rotatebox{90}{$WN(\frac{\pi}{2},\frac{3}{4})$}}&50&$4.04\cdot10^{-4}$&$5.00\cdot10^{-4}$&$3.68\cdot10^{-4}$&$4.08\cdot10^{-4}$&$3.71\cdot10^{-4}$&$3.78\cdot10^{-4}$&$3.59\cdot10^{-4}$&$3.88\cdot10^{-4}$&7.31&7.81\\
    &100&$1.86\cdot10^{-4}$&$6.74\cdot10^{-4}$&$1.67\cdot10^{-4}$&$2.26\cdot10^{-4}$&$1.90\cdot10^{-4}$&$1.89\cdot10^{-4}$&$1.62\cdot10^{-4}$&$1.74\cdot10^{-4}$&9.05&10.1\\
    &200&$1.74\cdot10^{-3}$&$9.65\cdot10^{-4}$&$8.29\cdot10^{-5}$&$1.41\cdot10^{-4}$&$9.29\cdot10^{-5}$&$9.66\cdot10^{-5}$&$6.77\cdot10^{-5}$&$7.15\cdot10^{-5}$&11.4&12.0\\
&500&$3.55\cdot10^{-3}$&$5.47\cdot10^{-3}$&$2.15\cdot10^{-4}$&$2.88\cdot10^{-4}$&$6.30\cdot10^{-5}$&$6.59\cdot10^{-5}$&$2.87\cdot10^{-5}$&$2.97\cdot10^{-5}$&15.3&15.7\\\hline\multirow{4}{*}{\rotatebox{90}{$WN(\frac{\pi}{2},\frac{9}{10})$}}&50&$2.44\cdot10^{-3}$&$2.76\cdot10^{-3}$&$1.48\cdot10^{-3}$&$1.53\cdot10^{-3}$&$2.05\cdot10^{-3}$&$2.08\cdot10^{-3}$&$1.32\cdot10^{-3}$&$1.35\cdot10^{-3}$&11.2&11.1\\
    &100&$6.94\cdot10^{-3}$&$1.16\cdot10^{-2}$&$1.41\cdot10^{-3}$&$1.67\cdot10^{-3}$&$1.51\cdot10^{-3}$&$1.51\cdot10^{-3}$&$8.05\cdot10^{-4}$&$8.15\cdot10^{-4}$&13.8&14.3\\
    &200&$2.72\cdot10^{-2}$&$3.12\cdot10^{-2}$&$2.14\cdot10^{-3}$&$2.30\cdot10^{-3}$&$1.13\cdot10^{-3}$&$1.13\cdot10^{-3}$&$4.54\cdot10^{-3}$&$4.57\cdot10^{-4}$&17.3&17.6\\
&500&$1.02\cdot10^{-1}$&$2.07\cdot10^{-1}$&$4.85\cdot10^{-3}$&$7.40\cdot10^{-3}$&$1.08\cdot10^{-3}$&$1.13\cdot10^{-3}$&$3.47\cdot10^{-4}$&$3.57\cdot10^{-4}$&23.4&27.8\\\hline
    \end{tabular}
    }
    \caption{MISE for density estimation based on rounded samples from different sampling distributions (leftmost column) and  different Berkson error distributions. In the left columns headed {\itshape param}, $m _{opt}$ has been calculated under a parametric assumption; the right columns headed {\itshape nonpar} use the non-parametric procedure.}
    \label{tab:BerksonRound}
\end{table}

\newpage

\begin{table}
    \centering
        \resizebox{\textwidth}{!}{
    \begin{tabular}{l|c|c|c|c|c|c|c|c}
    &&\multicolumn{4}{|c|}{$\MISE(F^{-\pi};\hat{F}_{m,n}^{-\pi})$}    
    &\multicolumn{1}{|c|}{Average}&\\\hline
     Distribution&$n$    &$m=5$ &$m=10$&$m=[\sqrt{n}]$ &$m=m_{OP}$&$m_{OP}$ &$m_{TH}$ \\\hline
     $VM(0,5)$    &50   & $2.37\cdot10^{-4}$ & $6.64\cdot10^{-5}$ & $1.14\cdot10^{-4}$ & $4.36\cdot10^{-5}$ &30.4 & 29.3\\
     $VM(\pi/2,5)$    &50   & $4.49\cdot10^{-4}$ & $9.18\cdot10^{-5}$ & $1.89\cdot10^{-4}$ & $4.18\cdot10^{-5}$ &39.0 &38.0 \\
     $VM(\pi,5)$    &50   & $5.84\cdot10^{-4}$ & $6.85\cdot10^{-4}$ & $5.72\cdot10^{-4}$ & $8.20\cdot10^{-4}$ &10.4 &9.03 \\
      $VM(0,1)$    &50  & $1.99\cdot10^{-4}$ & $2.38\cdot10^{-4}$ & $2.15\cdot10^{-4}$ & $2.40\cdot10^{-4}$ &8.79 & 7.78\\
      $VM(\pi/2,1)$    &50  & $3.63\cdot10^{-4}$ & $3.08\cdot10^{-4}$ & $3.18\cdot10^{-4}$ & $3.57\cdot10^{-4}$ &12.1 &11.3\\
      $VM(\pi,1)$    &50  & $3.73\cdot10^{-4}$ & $5.93\cdot10^{-4}$ & $4.72\cdot10^{-4}$ & $5.05\cdot10^{-4}$ &6.24 &5.19 \\
     $Mix(VM(0,5),VM(\pi/2,1),0.5)$    &50  & $1.91\cdot10^{-4}$ & $1.80\cdot10^{-4}$ & $1.77\cdot10^{-4}$ & $1.93\cdot10^{-4}$ &10.6 & 11.5
     \\
     $Mix(VM(0,5),VM(\pi/2,5),0.5)$    &50  &$1.80\cdot10^{-4}$ & $1.29\cdot10^{-4}$ & $1.41\cdot10^{-4}$ & $1.36\cdot10^{-4}$ &15.6 & 17.8
     \\
     $Mix(VM(0,5),VM(\pi/2,1),0.2)$    &50 &$2.23\cdot10^{-4}$ & $2.15\cdot10^{-4}$ & $2.09\cdot10^{-4}$ & $2.49\cdot10^{-4}$ &9.97 & 9.73
     \\
     $Mix(VM(0,5),VM(\pi/2,1),0.8)$    &50  & $2.03\cdot10^{-4}$ & $1.09\cdot10^{-4}$ & $1.37\cdot10^{-4}$ & $1.04\cdot10^{-4}$ &17.7 & 19.6
     \\
     $Mix(VM(0,1),VM(\pi/2,1),0.5)$    &50  &$2.04\cdot10^{-4}$ & $2.54\cdot10^{-4}$ & $2.23\cdot10^{-4}$ & $2.55\cdot10^{-4}$ &7.28 & 6.80\\
     $Mix(VM(0,1),VM(\pi/2,1),0.2)$    &50  & $2.42\cdot10^{-4}$ & $2.33\cdot10^{-4}$ & $2.27\cdot10^{-4}$ & $2.64\cdot10^{-4}$ &9.42 & 9.15
     \\
     $Mix(VM(0,5),VM(\pi,5),0.5)$   &50   &$2.20\cdot10^{-4}$ & $3.10\cdot10^{-4}$ & $2.52\cdot10^{-4}$ & $2.44\cdot10^{-4}$ &3.77 &6.69 
     \\\hline 
          
     $VM(0,5)$    &200   & $1.68\cdot10^{-4}$ & $2.11\cdot10^{-5}$ & $9.23\cdot10^{-6}$ & $2.40\cdot10^{-6}$ &82.9  & 82.0\\
     $VM(\pi/2,5)$    &200   & $3.62\cdot10^{-4}$ & $4.14\cdot10^{-5}$ & $1.64\cdot10^{-5}$ & $2.86\cdot10^{-6}$ &110  &109 \\
     $VM(
     \pi,5)$    &200   & $2.16\cdot10^{-4}$ & $7.39\cdot10^{-5}$ & $6.88\cdot10^{-5}$ & $7.76\cdot10^{-5}$ &22.2  &21.1 \\
      $VM(0,1)$    &200  & $2.10\cdot10^{-5}$ & $1.52\cdot10^{-5}$ & $1.50\cdot10^{-5}$ & $1.60\cdot10^{-5}$ &18.3 & 17.5\\
       $VM(\pi/2,1)$    &200  & $8.94\cdot10^{-5}$ & $3.97\cdot10^{-5}$ & $3.25\cdot10^{-5}$ & $3.01\cdot10^{-5}$ &28.3 & 27.8\\
        $VM(\pi,1)$    &200  & $3.55\cdot10^{-5}$ & $3.75\cdot10^{-5}$ & $4.26\cdot10^{-5}$ & $4.14\cdot10^{-5}$ &11.0 &10.2 \\
     $Mix(VM(0,5),VM(\pi/2,1),0.5)$    &200  & $4.14\cdot10^{-5}$ & $1.82\cdot10^{-5}$ & $1.57\cdot10^{-5}$ & $1.54\cdot10^{-5}$ &23.5 & 28.2
     \\
     $Mix(VM(0,5),VM(\pi/2,5),0.5)$    &200  &$6.26\cdot10^{-5}$ & $1.78\cdot10^{-5}$ & $1.30\cdot10^{-5}$ & $1.04\cdot10^{-5}$ &38.6 & 46.9
     \\
     $Mix(VM(0,5),VM(\pi/2,1),0.2)$    &200 &$5.37\cdot10^{-4}$ & $2.74\cdot10^{-5}$ & $2.39\cdot10^{-5}$ & $2.42\cdot10^{-5}$ &21.5 & 23.2
     \\
     $Mix(VM(0,5),VM(\pi/2,1),0.8)$    &200  & $8.87\cdot10^{-5}$ & $1.89\cdot10^{-5}$ & $1.17\cdot10^{-5}$ & $7.05\cdot10^{-6}$ &44.7 & 52.3
     \\
     $Mix(VM(0,1),VM(\pi/2,1),0.5)$    &200  &$2.08\cdot10^{-5}$ & $1.71\cdot10^{-5}$ & $1.81\cdot10^{-5}$ & $2.05\cdot10^{-5}$ &14.5 & 14.7\\
     $Mix(VM(0,1),VM(\pi/2,1),0.2)$    &200  & $5.04\cdot10^{-5}$ & $2.77\cdot10^{-5}$ & $2.51\cdot10^{-5}$ & $2.70\cdot10^{-5}$ &21.0 & 21.5
     \\
    $Mix(VM(0,5),VM(\pi,5),0.5)$  &200   &$4.09\cdot10^{-5}$ & $2.66\cdot10^{-5}$ & $2.87\cdot10^{-5}$ & $7.57\cdot10^{-5}$ &3.72 &14,4 
     \\\hline\hline
    \end{tabular}
    }
    \caption{MISE for CDF estimation with origin $\theta_0 = -\pi$ for (mixtures of) von Mises distributions}
    \label{tab:MixVMCDF}
\end{table}

\newpage

\begin{table}[!ht]
    \centering
        \resizebox{\textwidth}{!}{
    \begin{tabular}{l|c|c|c|c|c|c|c|c|c}
    &&\multicolumn{4}{|c|}{$\MISE(F^{\theta_0};\hat{F}_{m,n}^{\theta_0})$}    
    &\multicolumn{2}{|c|}{Average}&TH\\\hline
     Distribution&$n$    &$m=5$ &$m=10$&$m=[\sqrt{n}]$ &$m=m_{OP}$&$m_{OP}$ &$\theta^{OP}_0$&$\theta^{OP}_0$ \\\hline
     $VM(0,5)$    &50   & $2.47\cdot10^{-4}$ & $6.88\cdot10^{-5}$ & $1.19\cdot10^{-4}$ & $4.27\cdot10^{-5}$ &30.4 & -3.14&-3.14\\
     $VM(\pi/2,5)$    &50   & $2.47\cdot10^{-4}$ & $6.88\cdot10^{-5}$ & $1.19\cdot10^{-4}$ & $4.27\cdot10^{-5}$ &30.4 & -1.56&-1.57\\
    $VM(\pi,5)$    &50   & $2.47\cdot10^{-4}$ & $6.88\cdot10^{-5}$ & $1.19\cdot10^{-4}$ & $4.27\cdot10^{-5}$ &30.4 & 0.01&0.00\\
      $VM(0,1)$    &50  & $2.42\cdot10^{-4}$ & $2.44\cdot10^{-4}$ & $2.39\cdot10^{-4}$ & $2.54\cdot10^{-4}$ &8.87 & -3.11&-3.14\\
      $VM(\pi/2,1)$    &50  & $2.42\cdot10^{-4}$ & $2.44\cdot10^{-4}$ & $2.39\cdot10^{-4}$ & $2.54\cdot10^{-4}$ &8.87 & -1.55&-1.57\\
      $VM(\pi,1)$    &50  & $2.42\cdot10^{-4}$ & $2.44\cdot10^{-4}$ & $2.39\cdot10^{-4}$ & $2.54\cdot10^{-4}$ &8.87 & 0.02&0.00\\
    $Mix(VM(0,5),VM(\pi/2,1),0.5)$    &50  & $2.70\cdot10^{-4}$ & $2.41\cdot10^{-4}$ & $2.49\cdot10^{-4}$ & $2.56\cdot10^{-4}$ &10.4 & -2.40&-2.52
     \\
     $Mix(VM(0,5),VM(\pi/2,5),0.5)$    &50  &$1.54\cdot10^{-4}$ & $1.28\cdot10^{-4}$ & $1.33\cdot10^{-4}$ & $1.35\cdot10^{-4}$ &14.1 & -2.36&-2.36
     \\
     $Mix(VM(0,5),VM(\pi/2,1),0.2)$    &50 &$2.28\cdot10^{-4}$ & $2.42\cdot10^{-4}$ & $2.32\cdot10^{-4}$ & $2.44\cdot10^{-4}$ &8.00 & -1.94&-1.98
     \\
     $Mix(VM(0,5),VM(\pi/2,1),0.8)$    &50  & $2.48\cdot10^{-4}$ & $1.28\cdot10^{-4}$ & $1.67\cdot10^{-4}$ & $1.16\cdot10^{-4}$ &18.4 & -2.70&-2.94
     \\
     $Mix(VM(0,1),VM(\pi/2,1),0.5)$    &50  &$3.03\cdot10^{-4}$ & $3.41\cdot10^{-4}$ & $3.20\cdot10^{-4}$ & $3.25\cdot10^{-4}$ &6.39 & -2.37&-2.36\\
     $Mix(VM(0,1),VM(\pi/2,1),0.2)$    &50  & $2.24\cdot10^{-4}$ & $2.45\cdot10^{-4}$ & $2.32\cdot10^{-4}$ & $2.43\cdot10^{-4}$ &7.18 & -1.86&-1.85
     \\
     $Mix(VM(0,5),VM(\pi,5),0.5)$    &50   &$5.25\cdot10^{-4}$ & $5.22\cdot10^{-4}$ & $5.14\cdot10^{-4}$ & $5.81\cdot10^{-5}$ &4.16 & 1.58\footnotemark&$\pm1.57$
     \\\hline
          
     $VM(0,5)$    &200   & $1.71\cdot10^{-4}$ & $2.17\cdot10^{-5}$ & $9.62\cdot10^{-6}$ & $2.72\cdot10^{-6}$ &83.3  & -3.14&-3.14\\
      $VM(\pi/2,5)$    &200   & $1.71\cdot10^{-4}$ & $2.17\cdot10^{-5}$ & $9.62\cdot10^{-6}$ & $2.72\cdot10^{-6}$ &83.3 & -1.57&-1.57\\
    $VM(\pi,5)$    &200   & $1.71\cdot10^{-4}$ & $2.17\cdot10^{-5}$ & $9.62\cdot10^{-6}$ & $2.72\cdot10^{-6}$ &83.3& 0.00&0.00\\
      $VM(0,1)$    &200  & $2.81\cdot10^{-5}$ & $2.02\cdot10^{-5}$ & $1.94\cdot10^{-5}$ & $2.00\cdot10^{-5}$ &18.3 & -3.13&-3.14\\
     $VM(\pi/2,1)$    &50  & $2.81\cdot10^{-5}$ & $2.02\cdot10^{-5}$ & $1.94\cdot10^{-5}$ & $2.00\cdot10^{-5}$ &18.3 & -1.58&-1.57\\
      $VM(\pi,1)$    &50  & $2.81\cdot10^{-5}$ & $2.02\cdot10^{-5}$ & $1.94\cdot10^{-5}$ & $2.00\cdot10^{-5}$ &18.3& 0.00&0.00\\$Mix(VM(0,5),VM(\pi/2,1),0.5)$    &200  & $3.62\cdot10^{-5}$ & $1.80\cdot10^{-5}$ & $1.57\cdot10^{-5}$ & $1.55\cdot10^{-5}$ &22.3 & -2.49&-2.52
     \\
     $Mix(VM(0,5),VM(\pi/2,5),0.5)$    &200  &$4.19\cdot10^{-5}$ & $1.45\cdot10^{-5}$ & $1.15\cdot10^{-5}$ & $1.00\cdot10^{-5}$ &34.0 & -2.36&-2.36
     \\
     $Mix(VM(0,5),VM(\pi/2,1),0.2)$    &200 &$2.58\cdot10^{-4}$ & $2.05\cdot10^{-5}$ & $2.02\cdot10^{-5}$ & $2.09\cdot10^{-5}$ &15.5 & -1.98&-1.98
     \\
     $Mix(VM(0,5),VM(\pi/2,1),0.8)$    &200  & $9.14\cdot10^{-5}$ & $2.01\cdot10^{-5}$ & $1.26\cdot10^{-5}$ & $7.42\cdot10^{-6}$ &45.1 & -2.88&-2.94
     \\
     $Mix(VM(0,1),VM(\pi/2,1),0.5)$    &200  &$1.89\cdot10^{-5}$ & $1.82\cdot10^{-5}$ & $1.90\cdot10^{-5}$ & $1.95\cdot10^{-5}$ &11.6 & -2.35&-2.36\\
     $Mix(VM(0,1),VM(\pi/2,1),0.2)$    &200  & $2.06\cdot10^{-5}$ & $1.74\cdot10^{-5}$ & $1.75\cdot10^{-5}$ & $1.83\cdot10^{-5}$ &14.0 & -1.85&-1.85
     \\
     $Mix(VM(0,5),VM(\pi,5),0.5)$    &200   &$6.78\cdot10^{-5}$ & $3.89\cdot10^{-5}$ & $3.57\cdot10^{-5}$ & $9.73\cdot10^{-5}$ &4.10 & 1.58\footnotemark[1]&$\pm1.57$
     \\\hline
    
    \end{tabular}
    }
    \caption{MISE for CDF estimation with estimated origin $\theta_0$ for (mixtures of) von Mises distributions\\ \footnotesize{*Since the theoretical minimum is not unique and occurs at two opposite numbers, the average of the absolute values is taken.}}
    \label{tab:MixVMCDFestimatedTheta}
  
\end{table}

\end{document}